\documentclass[a4paper,10pt]{article}
\usepackage[utf8]{inputenc}
\usepackage{amsmath,amsthm, amssymb, amsfonts, amsbsy}

\usepackage{enumerate}
\usepackage[english]{babel}
\usepackage[numbers]{natbib}

\usepackage[a4paper, total={6in, 8in}]{geometry}
\usepackage{setspace}
\usepackage{thmtools}

\providecommand{\keywords}[1]{  \vspace{2ex}
\noindent\textbf{\textit{Keywords:}} #1
  \vspace{2ex}
}

\usepackage[mathscr]{euscript}
\usepackage{xspace}
\usepackage{graphicx}\usepackage{subcaption}
\usepackage{color}
\usepackage{enumerate}
\usepackage{enumitem}
\usepackage[english]{babel}
\usepackage{verbatim}
\usepackage{tikz}\usetikzlibrary{arrows.meta} \usepackage{arydshln} \usepackage{booktabs}\usepackage{natbib}
 \bibpunct[, ]{(}{)}{,}{a}{}{,}     
\usepackage{dsfont}

\usepackage[draft]{fixme}
\fxsetup{
  layout=marginnote,
  marginface=\normalfont\tiny,
  envface=,
  inlineface=,
  innerlayout=noinline,
}
\usepackage[textsize=tiny]{todonotes}

\usepackage[ruled, vlined, linesnumbered, resetcount]{algorithm2e}

\SetKw{Break}{break}
\SetKw{True}{true}
\SetKw{False}{false}
\DontPrintSemicolon

\usepackage{authblk}
\newenvironment{cr-problem}
{\vspace{10pt}\noindent{\bfseries\sffamily Convex Recoloring Problem (\CR).}\itshape\par\noindent}
{\vspace{10pt}}

\newcommand{\R}{\mathbb{R}}
\newcommand{\B}{\mathbb{B}}

\newcommand{\Q}{\mathbb{Q}}

\newcommand{\NP}{\mathsf{NP}}
\newcommand{\Pclass}{\mathsf{P}}

 \let\mathscr\relax\usepackage[scr]{rsfso}

\usepackage{hyperref}

\usepackage{perf}
\usepackage{pgfplots}

\usepackage{pgfplotstable}  \pgfplotsset{compat=1.18}   \usepgfplotslibrary{statistics}

\usepgfplotslibrary{colorbrewer}
\pgfplotsset{compat = newest} 
\usetikzlibrary{pgfplots.statistics, pgfplots.colorbrewer} 
\pgfplotsset{
    cycle list/Set1-9,     every axis plot/.append style={mark=*}
}

\usetikzlibrary{fit,backgrounds}\newtheorem{theorem}{Theorem}
\newtheorem{proposition}[theorem]{Proposition}

\newtheorem{lemma}[theorem]{Lemma}

\newtheorem{example}[theorem]{Example}

\providecommand{\Halmos}{\ensuremath{\square}}
\newcommand{\colors}{\mathscr{C}}

\newcommand{\CR}{\textsc{cr}}
\newcommand{\FR}{\mathcal{F}_{\textsc{r}}}
\newcommand{\FS}{\mathcal{F}_{\textsc{s}}}
\newcommand{\FC}{\mathcal{F}_{\textsc{c}}}
\newcommand{\FF}{\mathcal{F}_{\textsc{f}}}

\usepackage{cleveref}

\begin{document}

\title{Convex Recoloring of General Graphs: Formulations, Polyhedra, and Computational Experiments}

\author{Boyue Lin, Phablo F. S. Moura, Roel Leus}
\affil{Research Center for Operations Research and Statistics, KU Leuven\\
\texttt{\{boyue.lin, phablo.moura, roel.leus\}@kuleuven.be}}
\date{\today}

\maketitle

\begin{abstract}
 A vertex coloring of a graph is convex if the vertices of each color induce a connected subgraph. In the convex recoloring problem~(\CR), the goal is to find a convex coloring while minimizing the weight of recolored vertices, i.e., vertices assigned a color different from their original one. This problem was originally motivated by the study of phylogenetic trees in bioinformatics and is $\NP$-hard even on paths. Most existing research focuses on trees, with only limited results available for general graphs. We advance the state of the art by developing exact solution methods for \CR\ on general graphs. In particular, we propose  four mixed-integer linear programming formulations, including a compact flow-based model and a representatives model, and design corresponding solution methods. We compare the polytopes associated with the linear relaxation of the proposed formulations. Computational experiments on benchmark instances and on new synthetic instances show that a branch-and-cut algorithm based on the representatives formulation performs best overall.
\end{abstract}

\keywords{convex recoloring, branch and price, branch and cut, polyhedra, graph connectivity.}

\section{Introduction}
Let $G=(V,E)$ be an undirected graph, where $V$ and $E$ denote its sets of vertices and edges, respectively.
A \emph{(partial) coloring} of $G$ is a function~$C \colon V \to \colors \cup \{\emptyset\}$ that assigns to each vertex either a color from a set of colors~$\colors$, or $\emptyset$, which represents the absence of color.
A vertex~$v$ in~$V$ is said to be \emph{uncolored} if~$C(v)= \emptyset$.
If every vertex in $V$ has a color, then the coloring is called a \emph{total coloring}. 
For each~$c \in \colors$, the \emph{color class} of $c$ (denoted by~$C^{-1}(c)$) is the set of vertices that are assigned color~$c$. Note that this notion of coloring differs from the classical vertex coloring, in which adjacent vertices must have different colors.

A \emph{colored graph} is a pair~$(G,C)$ consisting of a graph~$G$ and a coloring $C$ of its vertices. A coloring~$C$ is said to be \emph{convex} if, for every~$c \in \colors$, the color class of~$c$ induces a connected subgraph of~$G$, that is, $G[C^{-1}(c)]$ is connected for every~$c \in \colors$. 
Given a colored graph~$(G,C)$, any other coloring $C'\colon V \to \colors \cup \{\emptyset\}$ is called a \emph{recoloring} of~$(G,C)$. A vertex $v\in V$ is \emph{recolored} by $C'$ if $C'(v)\neq C(v)$ and $C(v)\neq \emptyset$. The convex recoloring problem is defined as follows. 

\begin{cr-problem} 
\leavevmode
\begin{description}
 \raggedright
    \item[\textbf{Instance:}] A connected graph $G$, a coloring $C\colon V \rightarrow \colors \cup \{\emptyset\}$, and a weight function $w\colon V \rightarrow \Q_{\ge}$.
    \item[\textbf{Find:}] A convex recoloring $C'$ of $(G, C)$.
    \item[\textbf{Goal:}] Minimize $\sum_{v \in R(C')} w(v)$, where $R(C') := \{ v \in V : C(v) \neq \emptyset \text{ and } C(v) \neq C'(v) \}$ is the set of vertices recolored by $C'$.
  \end{description}
\end{cr-problem}

There are several applications of convex recoloring in phylogenetic trees~\citep{goldberg1996minimizing}, wireless communication systems~\citep{Campelo2021Heuristics}, protein-protein interaction networks~\citep{chor2007connected,vazquez2003global,schwikowski2000network}, and communication and transportation networks~\citep{calvert1997modeling,KAMMER2012The}.
In each of these domains, ensuring connectivity of vertices with the same attribute captures important structural or functional requirements. This naturally gives rise to formulations that explicitly enforce connectivity within each color class, which lies at the core of the models developed in this paper.
In phylogenetics, for instance, a convex coloring corresponds to a perfect phylogeny, and the recoloring objective can be viewed as a measure of distance from perfection. In protein-protein interaction networks, convex coloring ensures that proteins associated with the same cellular function or location induce connected subgraphs. 

Convex recoloring has been
studied primarily on trees due to its origin.
\citet{MORAN2008Convex} introduced this problem to investigate phylogenetic trees and proposed a dynamic programming algorithm for trees. Convex recoloring is known to be $\NP$-hard even on colored paths where each color appears at most twice~\citep{kanj2009convex, Moran2011}.
Several approximation algorithms are known for trees and paths. \cite{MORAN20071078} present a 2-approximation algorithm on paths. For any~$\varepsilon >0$, \cite{KAMMER2012The} propose a $(2+\varepsilon)$-approximation algorithm for bounded-treewidth graphs. \cite{lima2014convex} devise a $3/2$-approximation for paths in which each color appears at most twice. 

There are also several hardness results for graphs containing cycles.
\cite{campelo2014hardness} show that unweighted $\CR$ is $\NP$-hard on planar graphs with maximum degree~4. 
Additionally, they prove that, for a positive constant~$c$, there is no $c \ln n$-approximation algorithm even for \mbox{2-colored} \mbox{$n$-vertex} bipartite graphs, unless $\Pclass=\NP$.
\cite{moura2020strong} show that, for any \mbox{$\varepsilon>0$}, unweighted $\CR$ cannot be approximated within a ratio of $n^{1-\varepsilon}$ on $n$-vertex bipartite graphs, unless~$\Pclass=\NP$.

Exact methods for trees have also been developed.
\cite{Chopra2019} design an extended compact formulation for trees, and \cite{chopra2017column} propose a column generation method to solve convex recoloring on trees. 
\cite{Campelo2016} devise an integer linear programming (ILP) formulation with minimal vertex cut inequalities for arbitrary graphs, but they implement this model in a branch-and-cut algorithm only for \CR\ on trees. 
More recently, \cite{Campelo2022} proposed an ILP formulation based on connected subgraphs for general graphs, whereas 
their branch-and-price algorithm is restricted to trees.
In this class of graphs, they show that the associated pricing problem can be solved in polynomial time. \cite{dantas2022heuristic} propose heuristics for \CR\ on general graphs.
Although many applications of convex recoloring naturally arise on graphs with cycles, the problem remains largely unexplored in this setting. We address this gap by proposing exact solution methods for \CR\ on general graphs.

 Our contributions are threefold. First, we introduce four mixed-integer linear programming formulations, including a compact flow formulation and a representatives formulation, and we
 strengthen
 existing models using a structural result that allows us to restrict the solution space without loss of optimality. 
 Second, we analyze the polyhedral structure of these formulations and compare the strength of their linear relaxations. Third, we design dedicated solution methods, including branch-and-price and branch-and-cut algorithms, which enable a systematic comparison of distinct algorithmic frameworks for \CR. We also conduct extensive computational experiments on benchmark instances and new synthetic instances inspired by real-world networks, showing that a branch-and-cut algorithm based on the representatives formulation achieves the best overall performance.

 The remainder of the paper is organized as follows. \Cref{sec:Convex recoloring models} introduces the four formulations. 
\Cref{sec:Polyhedral comparison} presents the polyhedral comparison.
\Cref{sec:Branch-and-price scheme for the subgraph formulation} develops 
a branch-and-price algorithm for the connected subgraph formulation. 
\Cref{sec:Computational experiments} reports the computational experiments.
\Cref{sec:Conclusion} concludes the paper and discusses directions for future research.

In what follows, $(G,C,w)$ denotes an instance of \CR\ consisting of graph~$G$ with $n$ vertices and $m$ edges, a coloring $C$ that uses a set~$\colors$ with~$k$ colors, and $w(v)$ denotes the weight of vertex $v\in V$.
Note that minimizing the sum of the weights of recolored vertices is equivalent to maximizing the sum of the weights of vertices that keep their initial colors.
In this paper, we study $\CR$ as a maximization problem, in line with the work on solution approaches for \CR\ in the literature.

\section{Convex recoloring formulations}
\label{sec:Convex recoloring models}
We propose four formulations of \CR\ that model convex recolorings in which every color class contains at least one vertex that retains its original color. 
Formally, a coloring $C'$ is called \emph{expanding} (with respect to $C$) if, for every color $c\in \colors$ used in~$C'$, there exists a vertex $v$ such that $C'(v)=C(v)=c$. 
The following lemma is a straightforward generalization of the corresponding result for the unweighted case due to~\citet{Moran2011},
 but it was never used before to design solving methods for~\CR.

\begin{lemma}
\label{lemma:optimal maintains one color}
Let $I$ be an instance of \CR. 
There exists an optimal expanding convex recoloring of~$I$. 
\end{lemma}

All the proofs are provided in Appendix~\ref{Appendix:proofs}. Throughout this section, we use~\Cref{lemma:optimal maintains one color} to restrict the solution space without affecting optimality.
In the first three formulations, for all $v \in V$ and $c \in \colors$, we define $w_{vc} =w(v)$ if $C(v) = c$, and $w_{vc}=0$ otherwise. 

\subsection{Flow formulation}
Existing formulations for \CR\ either apply only to trees~\citep{Chopra2019} or involve exponentially many variables or constraints~\citep{Campelo2016,Campelo2022}.
We propose a compact flow-based formulation
for \CR\ on general graphs.
To the best of our knowledge, this is the first such formulation.

In this model, network flows are produced by a newly created vertex (source) and consumed by the vertices of the original graph.
For each color, the corresponding arborescence (i.e., a tree with all edges oriented from the source toward the leaves) consists of the source node and the vertices assigned that color, which act as sinks, each consuming exactly one unit of flow.
The idea of sending one unit of flow from the source to each vertex to enforce connectivity is inspired by the single commodity flow model for the Minimum Spanning Tree Problem~\citep{magnanti1995optimal} and by the flow formulation for the Balanced Connected $k$-Partition Problem~\citep{miyazawa2021partitioning}.

We first construct a directed graph from the undirected input graph~$G$ as follows. 
Let~$D$ be a directed graph with vertex set~$V\cup\{s\}$ and arc set 
\(\{(u,v), (v,u) \colon \{u,v\} \in E\} \cup \{(s,v) \colon v \in V\}.\)
In other words, $D$ is obtained from~$G$ by replacing every edge~$\{u,v\} \in E$ by the arcs~$(u,v)$ and~$(v,u)$, and adding a new vertex~$s$ with outgoing arcs to all vertices in~$V$.
For every $S \subseteq V(D)$, let us denote $A(S)$ as the set of arcs in $D$ with both endpoints in $S$.
Moreover, denote by $\delta^-(S)$ the set of arcs in $D$ with tail in~$V(D)\setminus S$ and head in $S$, that is, 
$\delta^-(S)=\{(u,v) \in A(D) : u\in  V(D)\setminus S \text{ and } v \in S\}$.
The set of arcs $\delta^+(S)$ is defined analogously. Let us define $\delta(S)=\delta^-(S)\cup\delta^+(S)$.

For each $a \in A(D)$ and $c \in \colors$, the model has a binary variable $y_{ac}$ that equals~$1$ if arc $a$ is selected to be in the arborescence connecting the vertices of color~$c$. 
The continuous variable~$f_{ac}$ represents the flow on arc $a$ corresponding to color~$c$.
For every~$v \in V(D)\setminus\{s\}$ and~$c \in \colors$, we define~$y(\delta^-(v),c) = \sum_{a \in \delta^-(v)} y_{ac} .$
A vertex $v \in V(D)\setminus\{s\}$ receives color $c$ in the recoloring if~$y(\delta^-(v),c)=1$.
Consider the following formulation of \CR: 
\begin{align} 
           \max \: &  \sum_{v \in V(D)\setminus\{s\}}  \sum_{c \in \colors} w_{vc} \, y(\delta^-(v), c) \label{for:flow_obj}\\
    \text{s.t.} \:& \sum_{a \in \delta^+(s)} y_{ac}  \leq 1  &     \forall c \in \colors \label{for:flow_1:new},\\
                \:& \sum_{c \in \colors} y(\delta^-(v), c)  \leq 1  &    \forall v \in V(D) \setminus \{s\} \label{for:flow_2:new},\\
                \: & f_{ac} \leq n \: y_{ac} &   \forall a \in A(D), c \in \colors \label{for:flow_4:new},\\
                \: & \sum_{a \in \delta^{-}(v)} f_{ac} -  \sum_{a \in \delta^+(v)} f_{ac}  =  y(\delta^-(v), c) &      \forall  v \in V(D)\setminus \{s \}, c \in \colors \label{for:flow_5:new},\\ 
               \: & y_{ac} \in \B &   \forall a \in A(D), c \in \colors \label{for:flow_6:new},\\ 
               \: & f_{ac} \in \R_\geq &  \forall a \in A(D), c\in \colors  \label{for:flow_7:new}.
\end{align}
Constraints~\eqref{for:flow_1:new} enforce that for each color, at most one chosen arc can leave the source $s$. Constraints~\eqref{for:flow_2:new} ensure each vertex in $V$ receives at most one incoming arc across all colors. 
Constraints~\eqref{for:flow_4:new} impose a capacity of $n$ on each arc. Constraints~\eqref{for:flow_5:new} require that each selected vertex consumes exactly one unit of flow 
for its assigned color. 
\begin{proposition}
\label{proposition:correctness_flow}
   The formulation~\eqref{for:flow_obj}--\eqref{for:flow_7:new} correctly models \CR.
\end{proposition}

We now exploit~\Cref{lemma:optimal maintains one color} in the foregoing formulation. 
For each color~$c \in \colors$, we may restrict the root of the corresponding arborescence to be a vertex that is initially colored~$c$. 
Indeed, in any expanding recoloring that uses color~$c$, at least one vertex retains color~$c$, and this vertex can be chosen as the root of the $c$-arborescence. 
Thus, we remove from the model the variables $f_{(s,v)c}$ and $y_{(s,v)c}$ for all $c\in \colors$ and $v\in V\setminus C^{-1}(c).$
We refer to the resulting reduced model as the flow formulation, denoted by~$\FF$.

The following inequalities~\eqref{for:flow_inequality} are valid for the flow formulation:
\begin{equation}
\label{for:flow_inequality}
y_{ac} \leq y(\delta^-(v), c) \qquad\qquad \forall  v \in V(D)\setminus \{s \}, a \in \delta^+(v), c \in \colors.
\end{equation}
These constraints state that, for each vertex other than the source, an outgoing arc of a given color can be selected only if the vertex has an incoming arc of the same color.

\subsection{Connected subgraph formulation}

The \textit{connected subgraph formulation} enforces connectivity directly through the decision variables by selecting connected subsets of vertices for each color.
Let $\mathcal{S}(G)$ denote the collection of all subsets of $V$ that induce connected subgraphs of $G$, and let $\eta = |\mathcal{S}(G)|$. When the graph $G$ is clear from the context, we simply write~$\mathcal{S}$.  
For a subset $V' \subseteq V$, let $\mathcal{S}_{\supset}(V') = \{ H \in \mathcal{S} : H \supseteq V' \}$ be the set of connected subsets that contain $V'$ and $\mathcal{S}_{\supset}(v) := \mathcal{S}_{\supset}(\{v\})$ for $v\in V$.  

A subgraph $H\in \mathcal{S}$ is called \emph{$c$-monochromatic} if all vertices in $H$ are assigned color $c$. 
The associated profit of making~$H$ $c$-monochromatic is
\(
w_{Hc} := \sum_{v \in H} w_{vc}.
\)
The following formulation of~\citet{Campelo2022} contains  a binary variable $l_{Hc}$ for every $H \in \mathcal{S}$ and $c \in \colors$, which equals~1 if and only if $H$ is $c$-monochromatic in the recoloring.
\begin{align}
\max \: & \sum_{H \in \mathcal{S}(G)} \sum_{c \in \colors} w_{Hc} \, l_{Hc} \label{for:connectSub_obj}\\
\text{s.t.} \: & \sum_{H \in \mathcal{S}_\supset (v)} \sum_{c \in \colors} l_{Hc} \leq 1 & \forall v \in V \label{for:connectSub_1},\\
               \: & \sum_{H \in \mathcal{S}(G)} l_{Hc} \leq 1 & \forall c \in \colors \label{for:connectSub_2},\\
               \: & l_{Hc} \in \{0,1\} & \forall H \in \mathcal{S}(G), \, c \in \colors.
\end{align}
Constraints~\eqref{for:connectSub_1} ensure that each vertex belongs to at most one connected subgraph, and therefore receives at most one color. Constraints~\eqref{for:connectSub_2} guarantee that each color is assigned to at most one connected subgraph. 

It is worth noting that, although the connected subgraph formulation in~\cite{Campelo2022} is described for general graphs, 
the associated pricing problem, implementation, and computational experiments in that work consider only tree instances of \CR. 
We now update the model of~\citeauthor{Campelo2022}  to incorporate the insight of \Cref{lemma:optimal maintains one color}. We restrict the solution space to expanding solutions as follows: if a connected subgraph $H$ is selected for color $c$, then $H$ must contain at least one vertex initially colored $c$. 
Thus, we can simply remove the variables $l_{Hc}$ for all $H \in \mathcal{S}(G)$ and $c \in \colors$  for which  $H\cap C^{-1}(c)=\emptyset$.
This reduced model is called the \emph{connected subgraph formulation} of~\CR, denoted by~$\FS$.

As the connected subgraph formulation may contain an exponential number of variables, we solve it using a branch-and-price approach,
for which the algorithmic details are developed in~\Cref{sec:Branch-and-price scheme for the subgraph formulation}.  In that section,
we show that the associated pricing problem reduces to the \emph{Maximum Weight Connected Subgraph problem} (MWCS),
which is known to be $\NP$-hard (see~\citealp{ideker2002discovering}). 
Given a vertex-weighted connected graph, MWCS is the problem of finding a connected subgraph of maximum total weight. 

\subsection{Vertex cut formulation}
\label{sec:Vertex cut formulation}
An alternative formulation enforces connectivity for each color class using vertex cut inequalities in the spirit of Menger's Theorem \citep{menger1927allgemeinen}.
A vertex set $Z$ is called a $(u,v)$-vertex cut if $u$ and~$v$ belong to different connected components of $G-Z$. 
Every such $(u,v)$-cut contains at least one vertex from each path connecting $u$ and $v$.
Let $\Gamma(u,v)$ denote the collection of inclusion-minimal $(u,v)$-vertex cuts.

\cite{CAMPELO2013233,Campelo2016} introduced the following formulation for \CR\ using vertex cuts.
For any vertex $v \in V$, let $N(v)$ denote the set of all vertices adjacent to $v$ in $G$, that is, $N(v)=\{u \in V :  \{u,v\} \in E\}$. 
For each vertex~$v \in V$ and color $c \in \colors$, we introduce a binary variable $h_{vc}$ such that $h_{vc}=1$ if and only if vertex $v$ receives color $c$ in the recoloring.
\begin{align}
\max \: & \sum_{v \in V} \sum_{c \in \colors} w_{vc} \, h_{vc} \label{for:vertexCut_obj} \\
\text{s.t.} \: & \sum_{c \in \colors} h_{vc} \leq 1 & \forall v \in V \label{for:vertexCut_1},\\
       \: & h_{uc} + h_{vc} - \sum_{z \in Z} h_{zc} \leq 1 & \forall c \in \colors,\,  u\in V, \, v \notin N(u), \, Z \in \Gamma(u,v)\label{for:vertexCut_2},\\
       \: & h_{vc} \in \{0,1\} & \forall v \in V, \, c \in \colors.
\end{align}
Constraints~\eqref{for:vertexCut_1} ensure that every vertex can receive only one color. 
Constraints~\eqref{for:vertexCut_2} enforce connectivity for each color class~$c\in\colors$ by requiring that every minimal vertex cut separating two nonadjacent vertices of color~$c$ contains a vertex of color~$c$.
The polyhedral structure of the convex hull of the binary vectors satisfying~\eqref{for:vertexCut_1} and~\eqref{for:vertexCut_2} is investigated by~\cite{CAMPELO2013233,Campelo2016}. 

We now modify this model using~\Cref{lemma:optimal maintains one color}. 
Specifically, we replace constraints~\eqref{for:vertexCut_2} with constraints~\eqref{for:vertexCut_lemma_retain-color} and~\eqref{for:vertexCut_lemma_connectivity}. 
\begin{align}
        \: & h_{vc} \le\sum_{u \in C^{-1}(c)} h_{uc} & \forall c \in \colors, v\in V\setminus C^{-1}(c) \label{for:vertexCut_lemma_retain-color},\\
       \: & h_{uc} + h_{vc} - \sum_{z \in Z} h_{zc} \leq 1 & \forall c \in \colors, u \in C^{-1}(c), v \notin N(u),   Z \in \Gamma(u,v). \label{for:vertexCut_lemma_connectivity}
\end{align}
Constraints~\eqref{for:vertexCut_lemma_retain-color} require that a color be used only if it is assigned to at least one vertex of that color in the input coloring.  Constraints~\eqref{for:vertexCut_lemma_connectivity} guarantee that any vertex assigned color~$c$ is connected to some vertex in $C^{-1}(c)$ that also receives color $c$.
Note that the number of inequalities~\eqref{for:vertexCut_lemma_connectivity} is approximately~$1/k$ times that of~\eqref{for:vertexCut_2}. 
Together with constraints~\eqref{for:vertexCut_1} and the binary restrictions on~$h$,
constraints~\eqref{for:vertexCut_lemma_retain-color} and
\eqref{for:vertexCut_lemma_connectivity} define the feasible region of the model.
The resulting \emph{vertex cut formulation} of~\CR, denoted by~$\FC$, is the problem of
maximizing~\eqref{for:vertexCut_obj} over this feasible region.

The vertex cut formulation is solved using a branch-and-cut approach in which the connectivity inequalities are separated dynamically via max-flow computations.  Similar to the separation of inequalities~\eqref{for:vertexCut_2} in \cite{CAMPELO2013233,Campelo2016}, inequalities~\eqref{for:vertexCut_lemma_connectivity} can be separated by solving a polynomial number of maximum-flow instances. 
In our implementation,
\Cref{lemma:optimal maintains one color} is exploited by restricting connectivity to vertices that retain their original color, thereby reducing the number  of pairs considered in the separation routine.
Implementation details of the separation procedures and the overall branch-and-cut algorithm are described in Section~\ref{sec:Computational setup}.

\subsection{Representatives formulation}
\label{sec:Representatives formulation}

All previous formulations use variables indexed by colors. We instead introduce a model based on vertex representatives that avoids color-indexed variables.
The formulation enforces connectivity through vertex cut inequalities, similarly to the vertex cut formulation. However, unlike that formulation, which assigns vertices to colors and requires vertices of the same color to be connected, the proposed model associates  
a representative vertex  with each color. Connectivity is then enforced between each vertex and its representative.

The idea of using vertex representatives in graph partitioning models was first introduced for the classic graph coloring problem by~\citet{campelo2004cliques}, and has since been applied to several combinatorial problems, including political districting~\citep{validi2022political}, the set-union bin packing problem~\citep{wahlen2025branch}, and parallel machine scheduling with conflicts~\citep{moura2025compact}.

For each color $c$, any vertex in $C^{-1}(c)$ may be chosen as the representative of that color, which leads to multiple symmetric solutions.
To eliminate this symmetry, we fix an order $\prec$ on the vertices and impose that a vertex represents its original color only if no preceding vertex with the same original color is assigned that color in the recoloring.

We are now ready to present the \emph{representatives formulation}~$\FR$\@.
For each $u\in V$, the set of potential representatives of $u$ is defined as
\(
\rho(u) := \{v\in V : \text{either } C(v)\neq C(u) \text{, or } C(v)=C(u) \text{ and } v\preceq u\}.
\)
\noindent For each \(u \in V\) and $v \in \rho(u)$, $\FR$ includes a binary variable \(x_{uv}\) that equals one if vertex~\(u\) receives color \(C(v)\) in the recoloring. In particular, for each \(v\in V\), \(x_{vv} = 1\)  if \(v\) is the representative of color~\(C(v)\), implying that $v$ keeps its original color. We define the weight function~$\hat{w}\colon V\times V \to \Q_\ge$ as: 
\begin{displaymath}
   \hat{w}_{uv} = \left\{
     \begin{array}{lr}
       w(u) & \text{ if } C(u)=C(v),\\
       0 & \text{ otherwise.}
     \end{array}
   \right.
\end{displaymath} 
The representatives formulation is as follows.
\begin{align}
\max \: & \sum_{u \in V} \sum_{v\in \rho(u)} \hat{w}_{uv}\,x_{uv} \label{for:repre_obj}\\
\text{s.t.}\: 
& \sum_{v \in \rho(u)} x_{uv} \leq 1 
&& \forall u \in V, \label{for:repre_1}\\
& x_{uv}\leq x_{vv} 
&& \forall u \in V, v \in \rho(u), \label{for:repre_2}\\
& \sum_{v \in C^{-1}(c)} x_{vv} \leq 1 
&& \forall c \in \colors, \label{for:repre_3}\\
& x_{uv}\leq \sum_{z \in Z: v \in \rho(z)} x_{zv} 
&& \forall  u \in V, v \in \rho(u)\setminus N(u), Z \in \Gamma(u,v), \label{for:repre_4}\\ 
& x_{uv} \in \{0,1\} 
&& \forall  u \in V, v \in \rho(u). \label{for:repre_integer}
\end{align}
\noindent Constraints~\eqref{for:repre_1} ensure that each vertex is represented by at most one other vertex, or is a representative. 
Constraints~\eqref{for:repre_2} guarantee that a vertex represents another only if it is itself a representative. 
Constraints~\eqref{for:repre_3} enforce that each color is represented by only one vertex. Constraints~\eqref{for:repre_4} ensure that if a vertex $v$ represents a vertex $u$, then it also represents at least one vertex in each minimal vertex cut separating $u$ and $v$.

Note that this formulation does not necessarily capture all optimal solutions to the convex recoloring problem. In particular, it excludes solutions in which a color class $c \in \colors$ is assigned exclusively to vertices that were not initially colored~$c$\@.
However, by \Cref{lemma:optimal maintains one color}, it includes at least one optimal solution.
\begin{proposition}
\label{proposition:correctness_repre}
    Formulation $\FR$ correctly models \CR.
\end{proposition}

The representatives formulation is also solved within a branch-and-cut framework, in which the connectivity constraints~\eqref{for:repre_4} are separated dynamically using max-flow computations, analogous to the treatment in the vertex cut formulation. Both constraints~\eqref{for:repre_2} and the exponentially many connectivity constraints \eqref{for:repre_4} are separated dynamically; for the latter, we identify the minimal vertex cuts with the same maximum-flow separation used for~\eqref{for:vertexCut_lemma_connectivity} in~$\FC$, now applied in the representative-variable space of~$\FR$. Implementation details are described in~\Cref{sec:Computational setup}.

\section{Polyhedral comparison} 
\label{sec:Polyhedral comparison}
This section compares the formulations introduced in Section~\ref{sec:Convex recoloring models} by analyzing the strength of their LP relaxations. The relationships among these relaxations are summarized in~\Cref{fig:lp-relaxation-relationships}. In the figure, an arc from formulation~$\mathcal A$ to formulation~$\mathcal B$ indicates that, under an objective-preserving projection onto the variable space of~$\mathcal B$, the polytope associated with the LP relaxation of~$\mathcal A$ is contained in
that of $\mathcal B$. 
We show that the connected subgraph formulation yields the strongest LP relaxation. Moreover, the representatives formulation dominates the vertex cut formulation.  The LP relaxations of the vertex cut and flow formulations are incomparable, in the sense that neither dominates the other.
  The only pairwise comparison left open is between the representatives formulation and the flow formulation.
  We return to this   question in~\Cref{sec:Computational results}, where we examine the corresponding LP bounds empirically.

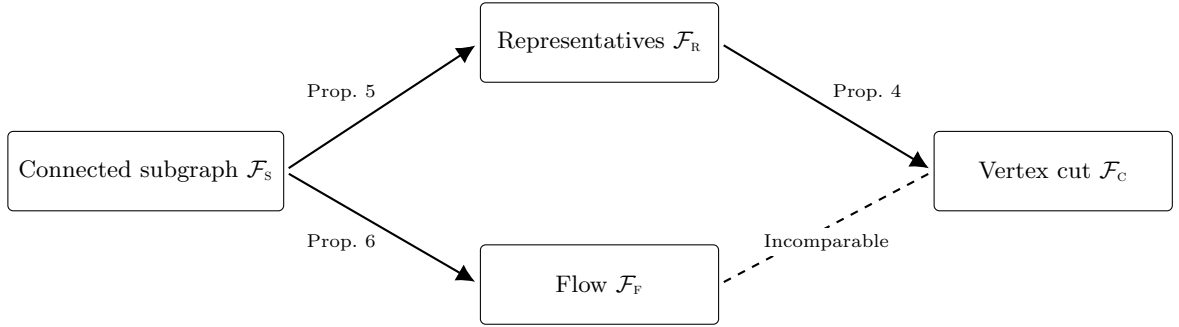
\begin{figure}[t]
\centering
\begin{tikzpicture}[
    formulation/.style={
        draw,
        rounded corners=2pt,
        align=center,
        minimum width=3.15cm,
        minimum height=1.05cm,
        inner sep=4pt,
        font=\small
    },
dominance/.style={
    -{Latex[length=7pt, width=7pt]},     thick,
    shorten >=2pt,
    shorten <=2pt
}, incomparable/.style={thick, dashed, shorten >=2pt, shorten <=2pt},
    label/.style={font=\scriptsize, fill=white, inner sep=2pt},
    figurebox/.style={
        rounded corners=6pt,
        inner sep=10pt
    }
]

    \node[formulation] (S) at (-6,0.25) {Connected subgraph $\FS$};
    \node[formulation] (R) at (0,1.95) {Representatives $\FR$};
    \node[formulation] (F) at (0,-1.25) {Flow  $\FF$};
    \node[formulation] (C) at (6,0.25) {Vertex cut $\FC$};

    \draw[dominance]
        (S.east) -- node[label, above left] {Prop.~\ref{prop:cs_in_repre}}
        (R.west);

    \draw[dominance]
        (R.east) -- node[label, above right] {Prop.~\ref{prop:repre_in_vertexcut_projection}}
        (C.west);

    \draw[dominance]
        (S.east) -- node[label, below left] {Prop.~\ref{prop:subgraph_in_flow_projection}}
        (F.west);

  \draw[incomparable, dashed]
    (F.east) --
    node[label, below] {Incomparable}
    (C.west);

        \begin{scope}[on background layer]
        \node[figurebox, fit=(S) (R) (F) (C)] {};
    \end{scope}
\end{tikzpicture}
\caption{Relationships among the LP relaxations of the four formulations. An arc $(\mathcal{A},\mathcal{B})$ indicates that $\mathcal{A}$ is stronger than $\mathcal{B}$. } 
\label{fig:lp-relaxation-relationships}
\end{figure}

Let $\mathcal{R}_{\textsc{r}}=\{ x \in \R^{n^2}_\ge :  x \text{ satisfies \eqref{for:repre_1}--\eqref{for:repre_4}} \}$ denote the polytope associated with the linear relaxation of $\FR$, and let $\mathcal{R}_{\textsc{c}} =\{ h \in \mathbb{R}_{\geq}^{n k} : h \text{ satisfies } \eqref{for:vertexCut_1},\eqref{for:vertexCut_lemma_retain-color},\eqref{for:vertexCut_lemma_connectivity}\}$ denote the corresponding polytope for $\FC$.
To relate the two formulations, we define a projection from the representatives space to the vertex cut space.
Let $ \xi\colon \mathcal{R}_{\textsc{r}}\to \R_{\geq}^{n k}$ be a linear function such that 
\begin{equation}
\label{map:repre_to_vc}
    \xi(x)_{uc} = \sum_{v \in C^{-1}(c)\cap \rho(u)} x_{uv}, \quad \forall u \in V\text{ and } c \in \colors.
\end{equation}

\begin{proposition}
\label{prop:repre_in_vertexcut_projection}
The projection $\xi$ maps $\mathcal{R}_{\textsc{r}}$ into $\mathcal{R}_{\textsc{c}}$, that is, $\{\xi(x) \in \R_\ge^{n k} : x \in \mathcal{R}_{\textsc{r}}  \}\subseteq \mathcal{R}_{\textsc{c}}$, and preserves the objective value.
\end{proposition}

The same argument also yields the following implication at the level of the
connectivity constraints. Under the projection~$\xi$,
constraints~\eqref{for:repre_2}--\eqref{for:repre_4} imply the vertex cut inequalities~\eqref{for:vertexCut_lemma_connectivity}; namely, 
\(
\left\{
\xi(x) :
x \in [0,1]^{n^2},\
x \text{ satisfies  \eqref{for:repre_2}--\eqref{for:repre_4}}
\right\}
\subseteq
\left\{
h \in [0,1]^{nk} :
h \text{ satisfies } \eqref{for:vertexCut_lemma_connectivity}
\right\}.
\)
This observation clarifies the relationship between the two classes of connectivity inequalities \eqref{for:repre_4} and~\eqref{for:vertexCut_lemma_connectivity}. Both classes are closely related to connectivity inequalities previously studied in the literature~\citep{wang2017imposing, validi2022imposing}. 
\color{black}

We now establish a similar relationship between the connected subgraph formulation and the representatives formulation.
Let~$\mathcal{R}_{\textsc{s}}$ be the polytope associated with the linear relaxation of $\FS$, that is,  $$\mathcal{R}_{\textsc{s}} =\{ l \in \mathbb{R}_{\geq}^{\eta k} : l \text{ satisfies } \eqref{for:connectSub_1},  \eqref{for:connectSub_2}, \text{ and } \ l_{Hc}= 0 \text{ for all } H \in \mathcal{S}(G) \text{ and } c \in \colors  \text{ with } H\cap C^{-1}(c)=\emptyset\}.$$
For each $H\in\mathcal{S}$ and $c\in\colors$, define $r(H,c) \in H\cap C^{-1}(c)$ such that $r(H,c) \preceq v$ for all $v \in H\cap C^{-1}(c)$ if $H\cap C^{-1}(c)\neq \emptyset$, and $r(H,c)=\emptyset$ otherwise.
In other words, $r(H,c)$ is the smallest vertex, according to the ordering~$\preceq$, among the vertices in $H$ whose initial color is $c$. 

We define a linear mapping $\phi\colon \mathcal{R}_{\textsc{s}}\to \mathbb{R}^{n^2}_{\geq}$ as follows:
\begin{equation}
\label{map:cs_to_repre}
    \phi(l)_{uv}
    =
    \sum_{\substack{H\in\mathcal{S}_{\supset}(u): r(H,C(v))=v}}
    l_{H,C(v)},
    \qquad
    \forall u\in V,\ v\in\rho(u).
\end{equation}

\begin{proposition}
\label{prop:cs_in_repre}
The projection $\phi$ maps $\mathcal{R}_{\textsc{s}}$ into $\mathcal{R}_{\textsc{r}}$, that is, $\{\phi(l)\in\mathbb{R}^{n^2}_{\geq}: l\in\mathcal{R}_{\textsc{s}}\} \subseteq\mathcal{R}_{\textsc{r}}$, and preserves the objective value.
\end{proposition}

We consider now the relaxations of the connected subgraph formulation to the flow formulation. 
Let $\mathcal R_{\textsc{f}}$ be the polytope associated with the linear relaxation of \(\FF\) which is defined as
\[
\mathcal R_{\textsc{f}}
= \left\{ (y,f) \in \R^{|A(D)|k}_\ge \times \R^{|A(D)|k}_\ge \;\middle|\;
\begin{aligned}
& (y,f) \text{ satisfies } \eqref{for:flow_1:new}-\eqref{for:flow_5:new}, \\
& f_{(s,v)c} = y_{(s,v)c} = 0 \quad \forall\, c \in \colors,\; v \in V \setminus C^{-1}(c)
\end{aligned}
\right\}.
\]

We next show a map $\lambda\colon \mathcal{R}_{\textsc{s}}\to \R^{|A(D)|k}_\ge\times\R^{|A(D)|k}_\ge$ which is defined based on the following subgraphs of $D$. 
For each $H \in \mathcal{S}$ and $c \in \colors$ with $H \cap C^{-1}(c) \neq \emptyset$, let $\vec H_c$
    be an (out-)arborescence in $D$ rooted at~$s$ obtained by taking an arbitrary (but fixed) spanning tree of $G[H]$, adding an edge $\{s,r\}$ where $r$ is any vertex in $H\cap C^{-1}(c)$, and orienting all edges away from $s$.
    Let us define $\mathcal{H}=\{\vec H_c : H \in \mathcal{S},c\in\colors,\ H\cap C^{-1}(c)\neq\emptyset\}$, and $\mathcal{H}_c(a)=\{\vec H_c \in \mathcal{H} : a \in A(\vec H_c)\}$ for all $a \in A(D)$ and $c \in \colors$.

    Given $\bar l \in \mathcal{R}_{\textsc{s}}$, 
    we define $(\bar y, \bar f) = \lambda(\bar l)$ as follows.
    For every $a \in A(D)$ and $c \in \colors$,
    let $\bar y_{ac} = \sum_{\vec H_c \in \mathcal{H}_c(a)} \bar l_{Hc}$ and $\bar f_{ac} = \sum_{\vec H_c \in \mathcal{H}_c(a)} n(\vec H_c, a) \bar l_{Hc}$, where $n(\vec H_c, a)$ is the number of vertices in the component of $\vec H_c - a$ that contains the head of $a$, that is, the vertex where arc $a$ enters.

\begin{proposition}
\label{prop:subgraph_in_flow_projection}
 The projection $\lambda\colon \mathcal{R}_{\textsc{s}}\to \R^{|A(D)|k}_\ge\times\R^{|A(D)|k}_\ge$ maps $\mathcal{R}_{\textsc{s}}$ into $\mathcal{R}_{\textsc{f}}$, that is, $  \{\lambda(l): l\in\mathcal{R}_{\textsc{s}}\} \subseteq\mathcal{R}_{\textsc{f}}$, and preserves the objective value.
\end{proposition}

We close this section by observing that all dominance relations established above are strict.
For illustration, for benchmark instance $080\_010\_000.gcc$ from the unweighted random graph test set, the LP bounds satisfy $\mathrm{LP}(\FF) > \mathrm{LP}(\FC) >\mathrm{LP}(\FR) > \mathrm{LP}(\FS)$, where $\mathrm{LP}(\mathcal{M})$ denotes the optimal value of the LP relaxation of formulation~$\mathcal{M}$. Moreover, ~\Cref{eg:VC_notin_flow} below provides an instance for which the vertex cut formulation is weaker than the flow model, showing that these are incomparable. 
\begin{example}
\label{eg:VC_notin_flow}
Let $G=(V,E)$ be a star graph with $V=\{v_0,v_1,v_2,v_3\}$ and $E=\bigl\{\{v_0,v_1\},\{v_0,v_2\},\{v_0,v_3\}\bigr\}$. 
The vertex weights are
$w(0)=3$ and $w(1)=w(2)=w(3)=1$, and the initial coloring is
$C(0)=a$ and $C(1)=C(2)=C(3)=b$.
An optimal solution keeps the color of $v_0$ and recolor two of the vertices $v_1, v_2, v_3$, yielding value~$4$.
The LP relaxations of the vertex cut and flow formulations have values $4.5$ and $4$, respectively. These values are attained by $h_{v_0a}=1$ and $h_{v_1b}=h_{v_2b}=h_{v_3b}=\frac{1}{2}$ for the vertex cut formulation, and by $y_{(s,v_0),a}=f_{(s,v_0),a}=1$ and $y_{(s,v_i),b}=f_{(s,v_i),b}=\frac{1}{3}$ for each $i\in \{1,2,3\}$ for the flow formulation, with all other variables equal to zero.
\color{black}
\end{example}
\color{black}

These results complete the hierarchy of LP relaxations among the proposed formulations and provide the basis for the computational comparison in the following sections.

\section{A branch-and-price procedure for $\FS$}
\label{sec:Branch-and-price scheme for the subgraph formulation}
In this section, we present a branch-and-price algorithm for the connected subgraph formulation, including the branching strategy and the solution of the associated pricing problems.

Let $\alpha \in\R_\geq^{n} $ and $\beta \in\R_\geq^{k}$ denote the dual variables associated with constraint sets \eqref{for:connectSub_1} and \eqref{for:connectSub_2}, respectively. The reduced cost of a variable $l_{Hc}$ is $-\beta_c+\sum_{v\in H}(w_{vc}-\alpha_v)$ for all $H \in \mathcal{S}$ and $c \in \colors$. 
Hence, for each color $c \in \colors$, the pricing problem in the root node of the branch-and-bound tree reduces to finding a  
connected subgraph $H \in \mathcal{S}$ that maximizes
\[
\sum_{v \in H} (w_{vc} - \alpha_v),
\]
subject to \(H \cap C^{-1}(c) \neq \emptyset\).
Thus, the pricing problem decomposes by color and can be solved independently for each color $c$. The resulting subproblem 
is equivalent to 
MWCS with vertex weights \(w_{vc} - \alpha_v\), and an additional covering constraint on \(C^{-1}(c)\) ensuring that the selected subgraph contains at least one vertex that is initially colored $c$.
The maximum-weight connected subgraph problem has been widely studied in the literature~\citep{alvarez2013maximum,rehfeldt2019combining}.
\color{black}

We branch on the implicit variables $h_{vc}=\sum_{H\in \mathcal{S}_\supset(v)} l_{Hc}$, which are defined for all $v\in V$ and $c\in\colors$, and whose interpretation coincides with that of the variables of the vertex cut formulation.
At each branching node, we create two child nodes by imposing either $h_{vc}\leq 0$ or $h_{vc}\geq 1$.
Let \(F_0\) be the set of pairs \((v,c)\) for which the constraint \(h_{vc} \le 0\) has been imposed along the path from the root to the current node (included). Similarly, let \(F_1\) denote the set of pairs \((v,c)\) for which \(h_{vc} \ge 1\) has been imposed.
Hence, every node may include the following branching constraints:
\begin{equation}
\label{eq:F0}
    h_{vc}=\sum_{H\in \mathcal{S}_\supset(v)} l_{Hc}\leq 0,\qquad \forall (v,c)\in F_0,
\end{equation}
\begin{equation}
\label{eq:F1}
    h_{vc}=\sum_{H\in \mathcal{S}_\supset(v)} l_{Hc}\geq 1,\qquad \forall (v,c)\in F_1.
\end{equation}
Constraint~\eqref{eq:F0} holds if and only if $l_{Hc}=0$ for all $H\in\mathcal{S}_\supset (v)$. Thus, for each $(v,c)\in F_0$, all columns corresponding to connected subgraphs that contain $v$ are forbidden. In the pricing problem for color~$c$, this is enforced by excluding vertex $v$ from the subgraph to be generated.
Similarly, constraint \eqref{eq:F1} requires that the selected column for color $c$ contains vertex $v$. Accordingly, for each $(v,c)\in F_1$, the pricing problem is restricted to connected subgraphs that include $v$. 
Moreover, all pairs $(v,c')$ with $c' \neq c$ are added to $F_0$, since a vertex assigned to color $c$ cannot be assigned to any other color.

Let $\gamma \in \R_\geq^{n\times k}$ be the dual variable associated with \eqref{eq:F1}. The reduced cost of~$l_{Hc}$ thus becomes: 
\begin{equation}
-\beta_c+\sum_{v\in H}(w_{vc}-\alpha_v)+\sum_{v\in H,(v,c)\in F_1}\gamma_{vc}.
\end{equation}
Now the pricing problem for a given color $c \in \colors$ is
\begin{equation}
-\beta_c+\max_{H \in \mathcal{S}} \left(\sum_{v\in H}(w_{vc}-\alpha_v)+\sum_{v\in H,(v,c)\in F_1}\gamma_{vc}\right),
\end{equation}
which corresponds to an instance of MWCS 
with vertex weights
\begin{displaymath}
    \bar w(v) = \left\{
     \begin{array}{lr}
       w_{vc}-\alpha_v+\gamma_{vc}, & \text{ if } (v,c)\in F_1,\\
       w_{vc}-\alpha_v, & \text{ otherwise.}
     \end{array}
   \right.
\end{displaymath} 
To solve the pricing problem, we employ the same branch-and-cut scheme developed for~$\FC$ restricted to a single color.

The branch-and-bound tree of the branch-and-price algorithm is explored using best-bound-first search. In the implementation, 
Gurobi solves the LP relaxations of the restricted master problems. It also solves the restricted master problem as a MIP to potentially improve the lower bound, both at the root node and at selected nodes where the bound has remained unchanged for several iterations. 

\section{Computational experiments}
\label{sec:Computational experiments}
\subsection{Computation Setup}
\label{sec:Computational setup}
All algorithms are implemented in C++ using the LEMON Graph Library version 1.3.1~\citep{dezsHo2011lemon}, and
Gurobi~11.0.2 as the MILP solver with all Gurobi parameters set to default unless otherwise specified.
The experiments are executed on a machine running MacOS~Sonoma~14.0 (64-bit), equipped with an Apple~M2 processor with 8 cores and 16\,GB of RAM\@. 
Gurobi can use up to 8 threads. The time limit for each instance is 1800 seconds.
An open-source implementation of the algorithms, together with the datasets used in the experiments, is available at \url{https://github.com/lynnac/convex-recoloring}.

We evaluate the algorithms associated with four formulations: the connected subgraph formulation~$\FS$,
the vertex cut formulation~$\FC$, the flow formulation~$\FF$, and the representatives formulation~$\FR$.
We denote the corresponding algorithms by S, C, F, and R, respectively.
Formulation~$\FF$ is solved directly using Gurobi, including the valid inequalities~\eqref{for:flow_inequality}.
For formulations~$\FC$ and~$\FR$,
we implement branch-and-cut algorithms.
For~$\FC$, the initial relaxation contains constraints~\eqref{for:vertexCut_1} and~\eqref{for:vertexCut_lemma_retain-color}; in preliminary tests, including the color-retention constraints~\eqref{for:vertexCut_lemma_retain-color} explicitly was more effective than separating them dynamically.
\color{black}
The separation algorithms of~\eqref{for:vertexCut_lemma_connectivity} and~\eqref{for:repre_4} proceed as follows. For a given color
(or representative), if all associated variables are integral, Algorithm~1 of \citet{fischetti2017thinning} is applied; otherwise, we use the separation algorithm of~\citet[Section~3.1]{miyazawa2021partitioning}.

The vertex cut formulation can also be strengthened by additional valid inequalities from the literature, including the generalized connectivity inequalities and multiway inequalities studied by~\citet{moura2026connected}. Moreover, these inequalities for~$\FC$ can be transferred to $\FR$ through the mapping in~\eqref{map:repre_to_vc}, and the corresponding separation routines can be adapted accordingly. We implemented and tested these additional inequalities for~$\FC$ as well as their mapped counterparts for~$\FR$.
Preliminary experiments showed that separating these inequalities did not improve performance on our benchmark instances. We therefore report results for C and R without these additional valid inequalities.
\color{black}

We evaluate the running times of the four algorithms using the
performance-profile methodology introduced by~\cite{dolan2002benchmarking}.
Let $A$ denote the set of algorithms under comparison and $P$ the set of instances. For each  $p \in P$ and $a \in A$, let
$t_{p,a}$ denote the running time of $a$ on $p$. The performance ratio of
algorithm $a$ on instance $p$ is defined as $r_{p,a} = \frac{t_{p,a}}{\min\{\,t_{p,b} : b \in A\,\}}$.
If algorithm $a$ fails to solve instance $p$, we set $r_{p,a} = +\infty$. For
each $a \in A$ and each performance factor $\tau \ge 1$, define
\[
\psi_a(\tau) = \frac{\lvert\{\,p \in P : r_{p,a} \le \tau\,\}\rvert}{\lvert P\rvert}.
\]
The function $\psi_a$ is the \emph{performance profile} of algorithm $a$.
Value $\psi_a(\tau)$ is the fraction of instances solved by algorithm~$a$ within a factor~$\tau$ of the best running time obtained by any algorithm. In particular, $\psi_a(1)$ is the fraction of instances on which $a$ is fastest.
\color{black}

To evaluate instances that are not solved to optimality within the time limit, we report the relative gap, defined as $\frac{\mathrm{UB}-\mathrm{LB}}{\mathrm{LB}}$, where UB and LB denote the upper bound and lower bound obtained by the algorithm for that instance, respectively. 
To compare solution quality across algorithms on these instances, we summarize the distribution of relative gaps using boxplots. Instances that are solved to optimality by all algorithms are omitted from these plots in order to better highlight differences on more challenging instances.
\color{black}

Computational experiments are conducted on random graphs and scale-free networks. For each class, we consider both unweighted and weighted instances.
The unweighted random graphs  are taken from the benchmark set of~\cite{crp-instances-page},
which is the only existing benchmark set for $\CR$ on general graphs. These instances are generated using the Erdős-Rényi random graph model. The graphs have no weights initially; we set all vertex weights to $1$.
The initial coloring is assigned greedily based on a proper-coloring rule \citep[see][]{bondy2011graph}, where no two neighbor vertices have the same color. 
Because solving all instances is computationally demanding, we select $10$ instances (the first $10$ from the original set) for each combination of $n\in \{60, 70, \ldots, 100\}$ and $p\in\{0.1,0.2,\ldots,0.5\}$, with $n$ being the number of vertices in the graph and $p$ the probability of occurrence of each edge.
Preliminary experiments confirmed that these sample instances are sufficient to capture the relative performance trends of the algorithms. We thus have a total of $250$ instances of this type.

The weighted random graph instances are also generated using the Erdős-Rényi random graph model, with vertex weights drawn uniformly from $\{1,\ldots,10\}$. In contrast to the unweighted random instances, the initial coloring is assigned uniformly at random. We generate three instances for each combination of $n \in \{120, 150, 180\}$, $k \in \{n/10, n/5, n/3\}$, and $p \in \{0.10, 0.20\}$, where $k$ denotes the number of colors in the initial coloring. This gives $54$ weighted random graph instances.

Scale-free networks are characterized by a power-law degree distribution. This characteristic has been observed in many large networks including those describing molecular interactions (e.g., protein-protein interactions)~\citep{barabasi2004network,pastor2003evolving}, certain social networks (e.g., from epidemiology)~\citep{pastor2002epidemic}, and communication networks (such as the world-wide web)~\citep{barabasi2000scale}. Many of these network structures are directly related to application domains in which the $\CR$ problem has been studied. 

To better understand algorithmic performance on instances that resemble these real-world applications, we generate unweighted and weighted scale-free instances using the Barabási-Albert model~\citep[see][]{barabasi1999emergence}. 
The weighted scale-free instances assign each vertex a weight equal to its degree, reflecting settings in which preserving the features of highly connected vertices is more important; for example, in protein networks, highly connected proteins are more likely to be essential~\citep[see][]{jeong2001lethality}. Details of the scale-free instance generation are given in Appendix~\ref{Appendix:scale-free generation}.

\color{black}

\subsection{Computational results}
\label{sec:Computational results}
\subsubsection{Effect of \Cref{lemma:optimal maintains one color}}

We first assess the impact of~\Cref{lemma:optimal maintains one color} on S, C, and F using a small set of unweighted large instances. This set consists of 30 random graph instances, with two instances for each combination of $n \in \{80, 90, 100\}$ and $p \in \{0.1, 0.2, \ldots, 0.5\}$, and 24 scale-free instances, with three instances for each combination of $n = 120$, $k \in \{n/15, n/10, n/5, n/3\}$, and $d \in \{2, 4\}$. These instances were chosen to be computationally challenging.
\Cref{fig:perf_and_gap_VC_lemma} shows that~\Cref{lemma:optimal maintains one color}~improves the performance of algorithm~C, leading to both faster solution times (left plot) and smaller optimality gaps (right plot).
For algorithms~S and~F, the inclusion of \Cref{lemma:optimal maintains one color} 
is also beneficial, although its effect on performance is less pronounced.
Therefore, all subsequent experiments are conducted using the model variants that incorporate \Cref{lemma:optimal maintains one color}.

\perfset{performance_new/runtime_data_VC_lemma.csv}

    \begin{figure}[t!]
        \centering
        \begin{subfigure}[t]{0.45\textwidth}
        \centering
        \begin{tikzpicture}
            \begin{axis}[
            height=6cm,
            extra x ticks={1},
            grid=both, 
            no marks, xlabel={Normalized time $\tau$}, ylabel={Proportion of  instances $\rho$}, xmin=1, xmax=20,
            label style={font=\footnotesize},
            tick label style={font=\footnotesize},    
            title style={font=\small},
            legend style={font=\footnotesize, at={(0.95,0.3)},anchor=east}]
                \addprofiles{2}{20}
                \legend{\textsc{without}, \textsc{with}, \textsc{r}, \textsc{f}}
            \end{axis}
        \end{tikzpicture}
        \end{subfigure}
        \hfill
        \begin{subfigure}[t]{0.45\textwidth}
        \centering
  \pgfplotstableread[col sep=comma]{performance_new/gap_data_VC_lemma_no_zero_rows.csv}\gapdata

    \begin{tikzpicture}
        \begin{axis}[
            y=0.3cm,
            x=0.5cm,
            boxplot/draw direction=y,
            xtick={1,2},
            xticklabels={\textsc{without}, \textsc{with}},
            ylabel={Gap (\%)},
            boxplot/whisker range=1000,        
            label style={font=\footnotesize},
            tick label style={font=\footnotesize},
            title style={font=\small},
            x tick label style={
                rotate=45,
                anchor=east}
        ]
            \addplot+[boxplot,fill,draw=black] table[y=0] {\gapdata};
            \addplot+[boxplot,fill,draw=black] table[y=1] {\gapdata};
        \end{axis}
    \end{tikzpicture}

        \end{subfigure}%
        \caption{Performance profiles and gaps for algorithm~C without (\textsc{without}) and with (\textsc{with})~\Cref{lemma:optimal maintains one color}.}
        \label{fig:perf_and_gap_VC_lemma}
    \end{figure}
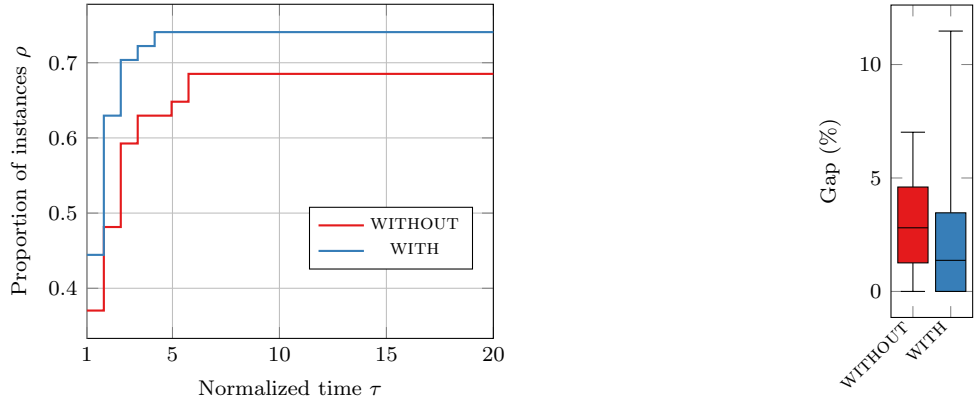

\subsubsection{Unweighted instances}
\label{subsub:unweightedresults}

Using the same subset of large unweighted instances as in the previous paragraph, we compare the LP bounds obtained from the four formulations. For every instance in this set, the upper bounds satisfy
\(\mathrm{LP}(\FF) \geq \mathrm{LP}(\FC) \geq \mathrm{LP}(\FR) \geq \mathrm{LP}(\FS).
\)
These empirical results are consistent with the polyhedral comparison presented in Section~\ref{sec:Polyhedral comparison}, where the connected subgraph formulation was shown to provide the strongest LP relaxation regardless of the instance characteristics. The results further suggest that both the representatives and vertex cut formulations offer a meaningful strengthening over the flow formulation in practice, as the latter consistently yields the weakest LP bounds across all tested instance classes.
\Cref{table:LP_comparison}~reports the average LP relaxation values, as well as their average percentages relative to the subgraph formulation, further illustrating the relative strength of  the four formulations.

\begin{table}[t]
\centering
\caption{Comparison of average LP relaxation values on large unweighted instances.}
\label{table:LP_comparison}
\begin{tabular}{lcccc}
\toprule
 & Subgraph & Representatives & Vertex cut & Flow \\
\midrule
LP relaxation & 72.23 & 73.18 & 77.76 & 82.51 \\
Percentage of subgraph & 100\% & 101.32\% & 107.65\% & 114.23\% \\
\bottomrule
\end{tabular}
\end{table}

\perfset{performance_new/runtime_data_random_all.csv}

    \begin{figure}[t!]
        \centering
        \begin{subfigure}[t]{0.45\textwidth}
        \centering
        \begin{tikzpicture}
            \begin{axis}[
            height=6cm,
            extra x ticks={1},
            grid=both, 
            no marks, xlabel={Normalized time $\tau$}, ylabel={Proportion of  instances $\rho$}, xmin=1, xmax=20,
            label style={font=\footnotesize},
            tick label style={font=\footnotesize},    
            title style={font=\small},
            legend style={font=\footnotesize, at={(0.95,0.3)},anchor=east}]
                \addprofiles{4}{20}
                \legend{\textsc{s}, \textsc{c}, \textsc{r}, \textsc{f}}
            \end{axis}
        \end{tikzpicture}
        \end{subfigure}
        \hfill
        \centering
        \begin{subfigure}[t]{0.2\textwidth}
        \centering
    \pgfplotstableread[col sep=comma]{performance_new/gap_data_random_all_no_zero_rows.csv}\gapdata

    \begin{tikzpicture}
        \begin{axis}[
            y=0.42cm,
            x=0.5cm,
            boxplot/draw direction=y,
            xtick={1,2,3,4},
            xticklabels={\textsc{s}, \textsc{c}, \textsc{r}, \textsc{f}},
            ylabel={Gap (\%)},
            boxplot/whisker range=1000,        
            label style={font=\footnotesize},
            tick label style={font=\footnotesize},
            title style={font=\small}
        ]
            \addplot+[boxplot,fill,draw=black] table[y=0] {\gapdata};
            \addplot+[boxplot,fill,draw=black] table[y=1] {\gapdata};
            \addplot+[boxplot,fill,draw=black] table[y=2] {\gapdata};
            \addplot+[boxplot,fill,draw=black] table[y=3] {\gapdata};
        \end{axis}
    \end{tikzpicture}
        \end{subfigure}%
          \hfill
        \begin{subfigure}[t]{0.2\textwidth}
        \centering
    \pgfplotstableread[col sep=comma]{performance_new/gap_data_random_all_no_flow.csv}\gapdata

    \begin{tikzpicture}
        \begin{axis}[
            y=1cm,
            x=0.5cm,
            boxplot/draw direction=y,
            xtick={1,2,3},
            xticklabels={\textsc{s}, \textsc{c}, \textsc{r}},
            ylabel={Gap (\%)},
            boxplot/whisker range=1000,        
            label style={font=\footnotesize},
            tick label style={font=\footnotesize},
            title style={font=\small}
        ]
            \addplot+[boxplot,fill,draw=black] table[y=0] {\gapdata};
            \addplot+[boxplot,fill,draw=black] table[y=1] {\gapdata};
            \addplot+[boxplot,fill,draw=black] table[y=2] {\gapdata};
        \end{axis}
    \end{tikzpicture}

        \end{subfigure}%
        \caption{Performance profiles and gaps for the unweighted random instances.}
        \label{fig:perf_and_gap_random_all}
    \end{figure}
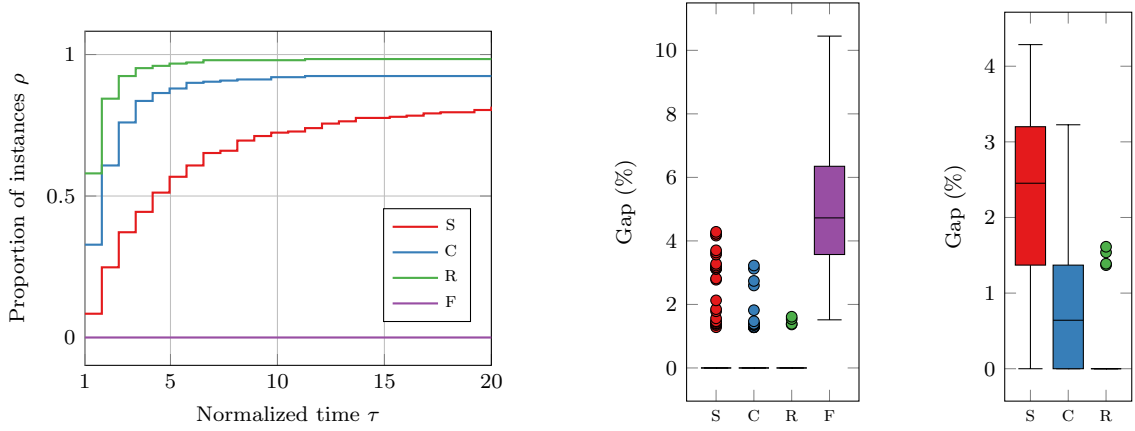

Out of the 250 instances
S, C, R, and F fail to solve 27, 17, 4, and 231 instances, respectively. The performance profiles  in 
\Cref{fig:perf_and_gap_random_all} show that algorithm~R is fastest for more than 58\% of 
the instances and clearly outperforms S, C, and~F overall. Algorithm~C is the second-best method in terms of running time, followed by S, while F performs significantly worse across the test set.
The figure also reports two gap boxplots: one including all four algorithms and one excluding F\@. The boxplots confirm the trend observed in the performance profiles. In particular, algorithm R achieves zero optimality gap on the vast majority of instances, while C performs second best, followed by~S\@. In contrast, F exhibits substantially larger gaps and variability. This behavior reflects the limited practical effectiveness of the flow formulation, likely due to its weaker LP bound, as observed above.  The second boxplot
excludes algorithm F and provides a clearer comparison among S, C, and R  by omitting instances that are solved to optimality by all three of these algorithms.  These plots highlight the robustness of algorithm R, whose gap distribution remains tightly concentrated around zero (the whiskers of R are at zero), whereas C and especially S exhibit larger variability.  The average optimality gaps across all 250 instances are 0.30\%, 0.12\%, 0.02\%, and 4.53\% for S, C, R, and F, respectively. These results indicate that the representatives formulation provides the strongest overall performance on this class of instances. 
Given its substantially weaker performance, we exclude formulation F from the remaining analysis, as it does not contribute meaningfully to the comparison among the best-performing methods.

To obtain a more detailed comparison of the three fastest algorithms, we next consider 135 unweighted scale-free network instances. The performance 
profiles depicted in \Cref{fig:perf_and_gap_scaleFree_all} show that C is the fastest  for over 51\% of the instances, followed by R (about  42\%), while S leads on only about 6\% of the instances. 
The boxplots show that R solves all instances to optimality, whereas S and C exhibit nonzero gaps on a subset of instances, with noticeably higher variability.
Out of the 135 instances, S and C fail to solve 13 and 10, respectively.
The average gaps across all instances are 0.23\%, 0.33\%, and 0.00\% for S, C, and R respectively.
Overall, these results indicate that algorithms C and R are particularly effective on scale-free networks, with R providing more robust performance across instances within the time limit, while C is sometimes faster.
This behavior can be partially explained by the structure of scale-free networks, which  tend to have small minimal vertex cuts.  This favors C and makes it the fastest method on many  easier instances. 
For the harder instances, however, the weaker LP bound of C becomes more limiting, leading to larger gaps.
In contrast, sparse scale-free graphs are structurally close to trees. Since the pricing problem in S is polynomial-time solvable on trees, the column generation is faster and produces smaller final gaps for difficult instances.
 \color{black}

\perfset{performance_new/runtime_data_PPI_all.csv}

    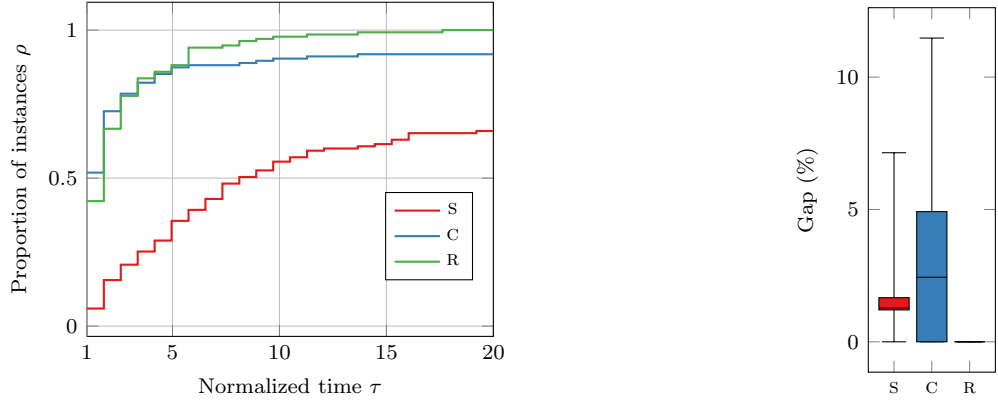
\begin{figure}[t!]
        \centering
        \begin{subfigure}[t]{0.45\textwidth}
        \centering
        \begin{tikzpicture}
            \begin{axis}[
            height=6cm,
            extra x ticks={1},
            grid=both, 
            no marks, xlabel={Normalized time $\tau$}, ylabel={Proportion of  instances $\rho$}, xmin=1, xmax=20,
            label style={font=\footnotesize},
            tick label style={font=\footnotesize},    
            title style={font=\small},
            legend style={font=\footnotesize, at={(0.95,0.3)},anchor=east}]
                \addprofiles{3}{20}
                \legend{\textsc{s}, \textsc{c}, \textsc{r}, \textsc{f}}
            \end{axis}
        \end{tikzpicture}
        \end{subfigure}
        \hfill
        \begin{subfigure}[t]{0.45\textwidth}
        \centering
   \pgfplotstableread[col sep=comma]{performance_new/gap_data_PPI_all_no_zero_rows.csv}\gapdata

    \begin{tikzpicture}
        \begin{axis}[
            y=0.35cm,
            x=0.5cm,
            boxplot/draw direction=y,
            xtick={1,2,3,4},
            xticklabels={\textsc{s}, \textsc{c}, \textsc{r}, \textsc{f}},
            ylabel={Gap (\%)},
            boxplot/whisker range=1000,        
            label style={font=\footnotesize},
            tick label style={font=\footnotesize},
            title style={font=\small}
        ]
            \addplot+[boxplot,fill,draw=black] table[y=0] {\gapdata};
            \addplot+[boxplot,fill,draw=black] table[y=1] {\gapdata};
            \addplot+[boxplot,fill,draw=black] table[y=2] {\gapdata};
        \end{axis}
    \end{tikzpicture}

        \end{subfigure}%
        \caption{Performance profiles and gaps for the unweighted scale-free instances.}
        \label{fig:perf_and_gap_scaleFree_all}
    \end{figure}

\subsubsection{Weighted instances}

On the 54 weighted random graph instances, algorithm S fails to produce a valid upper bound within the time limit for three instances. For these cases, we use the trivial upper bound obtained by recoloring all vertices in the graph.
The results in~\Cref{fig:perf_and_gap_weighted_random} show a different pattern compared to the unweighted case. Algorithm C is the fastest method on over $90\%$ of the instances, while algorithm R ranks  second in running time (it is the fastest method for around $8\%$ of the instance set) and 
maintains small optimality gaps. Algorithm~S is rarely the fastest method.
More precisely, the numbers of instances not solved by S, C, and~R are 8, 1, and~3, respectively. The average final gaps across all 54 instances are 1.10\%, 0.02\%, and 0.03\% for S, C, and R, respectively. 
This slight shift in performance compared to Section~\ref{subsub:unweightedresults} suggests that the vertex cut formulation is particularly effective when vertex weights introduce heterogeneity in the problem structure. 
The weights help distinguish among competing recolorings and guide the search more effectively,
while the relatively simple color-indexed assignment structure allows the separation routines to exploit this information effectively. 
\color{black}
Overall, on weighted random graphs, the vertex cut formulation provides the strongest performance among the three algorithms.

\perfset{performance_new/runtime_data_weighted_random.csv}

    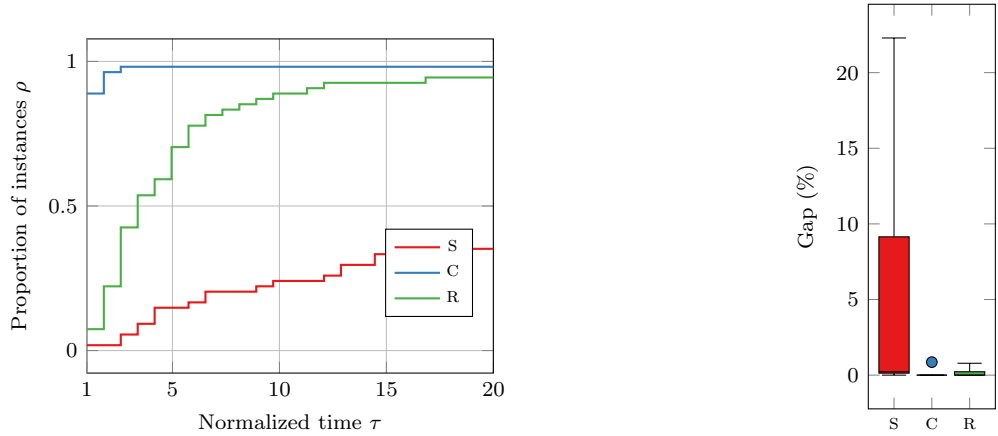
\begin{figure}[t!]
        \centering
        \begin{subfigure}[t]{0.45\textwidth}
        \centering
        \begin{tikzpicture}
            \begin{axis}[
            height=6cm,
            extra x ticks={1},
            grid=both, 
            no marks, xlabel={Normalized time $\tau$}, ylabel={Proportion of  instances $\rho$}, xmin=1, xmax=20,
            label style={font=\footnotesize},
            tick label style={font=\footnotesize},    
            title style={font=\small},
            legend style={font=\footnotesize, at={(0.95,0.3)},anchor=east}]
                \addprofiles{3}{20}
                \legend{\textsc{s}, \textsc{c}, \textsc{r}, \textsc{f}}
            \end{axis}
        \end{tikzpicture}
        \end{subfigure}
        \hfill
        \begin{subfigure}[t]{0.45\textwidth}
        \centering
   \pgfplotstableread[col sep=comma]{performance_new/gap_data_weighted_random_no_zero_rows.csv}\gapdata

    \begin{tikzpicture}
        \begin{axis}[
            y=0.2cm,
            x=0.5cm,
            boxplot/draw direction=y,
            xtick={1,2,3,4},
            xticklabels={\textsc{s}, \textsc{c}, \textsc{r}, \textsc{f}},
            ylabel={Gap (\%)},
            boxplot/whisker range=1000,        
            label style={font=\footnotesize},
            tick label style={font=\footnotesize},
            title style={font=\small}
        ]
            \addplot+[boxplot,fill,draw=black] table[y=0] {\gapdata};
            \addplot+[boxplot,fill,draw=black] table[y=1] {\gapdata};
            \addplot+[boxplot,fill,draw=black] table[y=2] {\gapdata};
        \end{axis}
    \end{tikzpicture}

        \end{subfigure}%
        \caption{Performance profiles and gaps for the weighted random instances.}

\label{fig:perf_and_gap_weighted_random}
    \end{figure}

For the weighted scale-free instances,
\Cref{fig:perf_and_gap_weighted_PPI_degree} shows that R is clearly the most effective method. The percentages of instances not solved by S and C are approximately $16\%$ and $13\%$, respectively, whereas R solves all instances to optimality within the time limit. Moreover, R is the fastest algorithm for more than $78\%$ of the test set, while C and S are fastest for approximately $18\%$ and $4\%$ of the instances, respectively. The average final gaps over all 135 instances are $3.09\%$, $1.01\%$, and $0.00\%$ for S, C, and R, respectively. These results confirm the robustness of the representatives formulation, which for the weighted scale-free case consistently solves all instances to optimality and dominates the other methods in both running time and solution quality.

\perfset{performance_new/runtime_data_weighted_PPI_degree.csv}

    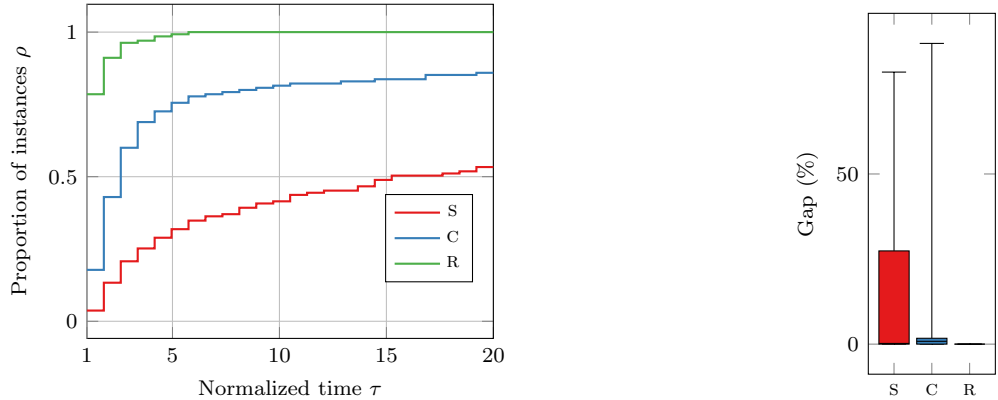
\begin{figure}[t!]
        \centering
        \begin{subfigure}[t]{0.45\textwidth}
        \centering
        \begin{tikzpicture}
            \begin{axis}[
            height=6cm,
            extra x ticks={1},
            grid=both, 
            no marks, xlabel={Normalized time $\tau$}, ylabel={Proportion of  instances $\rho$}, xmin=1, xmax=20,
            label style={font=\footnotesize},
            tick label style={font=\footnotesize},    
            title style={font=\small},
            legend style={font=\footnotesize, at={(0.95,0.3)},anchor=east}]
                \addprofiles{3}{20}
                \legend{\textsc{s}, \textsc{c}, \textsc{r}, \textsc{f}}
            \end{axis}
        \end{tikzpicture}
        \end{subfigure}
        \hfill
        \begin{subfigure}[t]{0.45\textwidth}
        \centering
   \pgfplotstableread[col sep=comma]{performance_new/gap_data_weighted_PPI_degree_no_zero_rows.csv}\gapdata

    \begin{tikzpicture}
        \begin{axis}[
            y=0.045cm,
            x=0.5cm,
            boxplot/draw direction=y,
            xtick={1,2,3,4},
            xticklabels={\textsc{s}, \textsc{c}, \textsc{r}, \textsc{f}},
            ylabel={Gap (\%)},
            boxplot/whisker range=1000,        
            label style={font=\footnotesize},
            tick label style={font=\footnotesize},
            title style={font=\small}
        ]
            \addplot+[boxplot,fill,draw=black] table[y=0] {\gapdata};
            \addplot+[boxplot,fill,draw=black] table[y=1] {\gapdata};
            \addplot+[boxplot,fill,draw=black] table[y=2] {\gapdata};
        \end{axis}
    \end{tikzpicture}

        \end{subfigure}%
        \caption{Performance profiles and gaps for the weighted scale-free instances.}

\label{fig:perf_and_gap_weighted_PPI_degree}
    \end{figure}

\subsubsection{Conclusion of the computational results} 

Overall, the results show that the representatives formulation performs best in most settings, followed by the vertex cut formulation, which is particularly effective on weighted random instances.
In particular, the experiments reveal a trade-off between speed and robustness, with the vertex cut formulation often achieving the shortest running times and the representatives formulation providing more consistent solution quality.

These findings suggest that LP bound strength alone does not fully determine practical performance. The representatives formulation appears to strike a favorable computational balance: although its LP relaxation is weaker than that of the connected subgraph formulation, the empirical difference between their LP bounds remains modest on the tested instances. At the same time, it avoids the overhead of branch-and-price and benefits from an ordering rule that reduces symmetry. 

In contrast,
although the connected subgraph formulation provides the strongest LP bounds, it is generally slower in practice due to the need to repeatedly solve $\NP$-hard pricing subproblems within 
the branch-and-price framework.
This creates a competitive disadvantage relative to branch-and-cut approaches, which can better exploit built-in solver features such as branching, cutting planes, heuristics, and node selection strategies.
The flow formulation performs worst overall, likely due to its large number of arc-color variables, continuous flow variables, and big-\(M\)-type capacity constraints, which collectively weaken the LP relaxation and increase computational effort.

\section{Conclusion and further research}
\label{sec:Conclusion}

In this work, we developed and evaluated exact solution methods for the convex recoloring problem on general graphs.  We proposed four mixed-integer linear programming formulations, including a novel compact flow formulation and a representatives formulation, and we strengthened the models using a structural result ensuring the existence of optimal expanding solutions. 
Each formulation was paired with a dedicated solution approach, including branch-and-price and branch-and-cut algorithms, enabling a systematic comparison of distinct algorithmic paradigms.

Overall, the representatives formulation emerges as the most effective approach in practice, offering the best balance between formulation strength and computational efficiency. The vertex cut formulation remains competitive, particularly on weighted random instances, where it frequently achieves the fastest running times. In contrast, although the connected subgraph formulation yields the strongest LP bounds, its branch-and-price algorithm is typically slower in practice, which is likely due to the complexity of the pricing problem. The flow formulation consistently performs worst, reflecting both its weaker LP relaxation and its relatively large formulation size.  Our results highlight that LP relaxation strength alone does not determine computational performance; rather, the interaction between formulation size, structure, and algorithmic framework plays a crucial role.

Future research could focus on strengthening the representatives formulation through additional valid inequalities and improving the efficiency of the branch-and-price framework for the connected subgraph formulation. Although our computational results suggest that the LP relaxation of the flow formulation is not stronger than that of the representatives formulation, the exact theoretical relationship between these two models remains an open theoretical question. Investigating hybrid approaches that combine features of different formulations also appears to be a promising direction.

\bibliographystyle{plainnat}
\bibliography{bibliography}

\newpage
\appendix

\section{Proofs}
\label{Appendix:proofs}
\subsection{Proof of \Cref{lemma:optimal maintains one color}}
\label{Appendix:proof of lemma:optimal maintains one color}
\begin{proof}{Proof}
Let $I=(G,C,w)$, and let $C'$ be an optimal convex recoloring of $I$.
Suppose $C'$ is not expanding. 
If $C'$ only uses one color, we can construct another recoloring $C''$ by recoloring all the vertices with one of the original colors used in $C$. 
This $C''$ is thus expanding, and its objective value cannot be worse than that of $C'$, so $C''$ is also optimal.

Suppose now that~$C'$ uses more than one color, and let $c$ be a color used in $C'$ such that no vertex $v \in V(G)$ satisfies $C(v)=C'(v)=c$. Since $G$ is connected, there exists an edge $(u,v)$ with $C'(u)=c$ and $C'(v)=d\neq c$. We can construct a new recoloring $C''$ by assigning color $d$ to all vertices with color $c$ in $C'$, leaving all other vertex colors unchanged. 
It is clear that $C''$ is convex, $R(C'') \subseteq R(C')$, and so $\sum_{v \in R(C'')} w(v)\leq \sum_{v \in R(C')} w(v)$. 
Hence, $C''$ is also optimal. 
If $C^{\prime\prime}$ is still not expanding, we repeat this process successively until $C''$ becomes expanding or until there is only one color used in the new recoloring.  In the latter case, we can revert to the first part of this proof to establish the lemma.  \hfill \Halmos
\end{proof}

\subsection{Proof of~\Cref{proposition:correctness_flow}}
\begin{proof}{Proof}
Let $C'\colon V\to \colors$ be a convex recoloring of $(G,C)$.
Let $\mathcal{T}=\{T_c\}_{c\in \colors}$ be a collection of pairwise disjoint arborescences in $D\setminus \{s\}$ where $T_c$ spans vertices $v\in V$ such that $C'(v)=c$, for every $c \in \colors$. It is clear that~$\mathcal{T}$ exists since $C'$ is a convex coloring.
We define a vector $(\bar y, \bar f) \in \B^{(2m+n) k} \times \R^{(2m+n) k}_{\ge}$ as follows.
The nonzero entries of $\bar y$ are precisely $\bar y_{ac}=1$ for all $c \in \colors$ and $a \in A(T_c)$, and $\bar y_{(sr)c}=1$ where~$r$ is the root of~$T_c$.
The nonzero entries of~$\bar f$ are as follows: $\bar f_{ac}$ equals the cardinality of the (maximal) sub-arborescence of $T_c-a$ that contains the head of $a$ for all $c \in \colors$ and $a \in A(T_c)$, and $\bar f_{(sr)c}=|V(T_c)|$, where~$r$ is the root of $T_c$.
It follows from the definition of $w$ that the value of~$(\bar y, \bar f)$ equals the sum of the weights of the vertices that have the same color in the initial coloring~$C$ and in the recoloring~$C'$.

Clearly, $(\bar y, \bar f)$ satisfies~\eqref{for:flow_1:new}.
Every vertex of $G$ belongs to at most one arborescence of $\mathcal{T}$, and therefore constraints~\eqref{for:flow_2:new} 
hold for $(\bar y, \bar f)$.
Note that the vector also satisfies~\eqref{for:flow_4:new} since,  for each $c \in \colors$,  $\bar f_{ac}$ is positive only if $a$ belongs to $A(T_c)$ or is an arc from $s$ to the root of $T_c$, and  its value is at most the cardinality of the arborescence.
Additionally, constraints~\eqref{for:flow_5:new} hold for $(\bar y, \bar f)$ because every vertex in~$T_c$ consumes exactly one unit of incoming flow, i.e., $y(\delta^-(v),c)=1$.

For the converse, we prove that every feasible solution of this formulation induces a convex recoloring of $G$.
Let $(\bar y, \bar f)$ be a feasible solution to~\eqref{for:flow_1:new}--\eqref{for:flow_7:new}.
For each color $c \in \colors$, define $T_c := D[\{a\in A(D): \bar y_{ac}=1\}]$ and $V_c := V(T_c)\setminus\{s\}$.
Because of~\eqref{for:flow_2:new}, the sets $V_c$, $c\in\colors$, are pairwise
disjoint and every vertex in $V_c$ has indegree at most one.
Constraints~\eqref{for:flow_4:new}, \eqref{for:flow_5:new}, and the nonnegativity of the flows guarantee that $T_c$ contains no directed cycle and that every nonempty component receives flow from $s$.
Together with~\eqref{for:flow_1:new}, this implies that $T_c$ is an arborescence for every~$c \in \colors$.
Therefore, $\{V_c\}_{c \in \colors}$ induces a convex recoloring of $G$ with an objective value equal to 
\begin{align*}
\sum_{c\in \colors} \sum_{v \in V_c : C(v)=c} w(v) = \sum_{c\in \colors}\sum_{v \in V : C(v)=c} w(v) \bar y(\delta^-(v),c)= \sum_{c\in \colors}\sum_{v \in V } w_{vc} \bar y(\delta^-(v),c). \tag*{\Halmos}
\end{align*}
\end{proof}

\subsection{Proof of~\Cref{proposition:correctness_repre}}
\begin{proof}{Proof}
Let $C'$ be an \emph{expanding} (with respect to $C$) convex recoloring.
For each $c \in \colors'$, let $v_c \in V$  be
the $\preceq$-minimal vertex in
$\{v\in V : C'(v)=C(v)=c\}.$
Define a vector $\bar x \in \B^{n^2}$ such that its nonzero entries are precisely $\bar x_{u v_c}=1$ for all $u \in V$ with $C'(u)=c$.
It is clear that $\bar x$ satisfies~\eqref{for:repre_1}--\eqref{for:repre_3}.
Moreover, the convexity of $C'$ implies that $\bar x$ satisfies constraints~\eqref{for:repre_4}, and so it is feasible to~$\FR$.

Let~$\bar x \in \B^{n^2}$ be a solution to $\FR$.
Note that each color in $\colors$ has at most one representative vertex by~\eqref{for:repre_3}, and define $\colors'= \{C(v) \in \colors : \bar x_{vv}=1\}$.
For each $v \in V$ such that $\bar x_{vv}=1$, define $H_c=\{u\in V :  \bar x_{uv}=1\}$ where $c = C(v)$.
Because of~\eqref{for:repre_1}, the sets in $\mathcal{H} := \{H_c\}_{c \in \colors'}$ are pairwise disjoint.
Consider $c \in \colors'$, and let $v \in V$ be the representative of $c$, that is, $\bar x_{vv}=1$ and $C(v)=c$. 
Suppose now that $H_c$ is not connected, and let $u \in H_c$ belong to a component that does not contain $v$.
Clearly, there is $Z \subseteq V\setminus H_c$ that is a minimal $(u,v)$-vertex cut.
Hence $\bar x_{uv} = 1 > 0 = \sum_{z \in Z} \bar x_{zv}$, which contradicts the fact that $\bar x$ satisfies~\eqref{for:repre_4}.
Consequently, each $H_c$ is connected, and therefore $\mathcal{H}$ induces an expanding (w.r.t. $C$) convex recoloring of value $\sum_{c \in \colors'} \sum_{u \in H_c : C(u)=c} w(u) =  \sum_{u \in V} \sum_{v\in \rho(u) : C(u)=C(v)} w(u) \bar x_{uv} = \sum_{u \in V} \sum_{v\in \rho(u) : C(u)=C(v)} \hat{w}_{uv} \bar x_{uv}$.
It follows from~\Cref{lemma:optimal maintains one color} that an optimal solution of~$\FR$ yields an optimal solution to \CR, and so the formulation is correct. \hfill \Halmos
\end{proof}

\subsection{Proof of~\Cref{prop:repre_in_vertexcut_projection}}
\begin{proof}{Proof}
    Consider $\bar x \in \mathcal{R}_{\textsc{r}}$, and define the vector $\bar h=\xi(\bar x)$.
    For each $u\in V$, we have $\sum_{c\in \colors} \bar h_{uc}=\sum_{c\in \colors}\sum_{v\in C^{-1}(c)\cap \rho(u)}\bar x_{uv}\le 1$ as $\bar x$ satisfies~\eqref{for:repre_1}.
    Thus $\bar h$ satisfies~\eqref{for:vertexCut_1}. 
    By \eqref{for:repre_2}, $\bar{x}_{vr}\leq \bar{x}_{rr}$ for all $ v\in V$ and $r \in \rho(v)$. 
    Thus, for each $c\in \colors,v\in V\setminus C^{-1}(c)$, $\bar{h}_{vc}=\sum_{r\in C^{-1}(c)\cap \rho(v)}\bar x_{vr}\leq \sum_{r\in C^{-1}(c)\cap \rho(v)}\bar{x}_{rr}\leq \sum_{r\in C^{-1}(c)}\bar{h}_{rc}$, where the last inequality holds because $\bar x_{rr}$ appears in $\bar{h}_{rc}$. Consequently, \eqref{for:vertexCut_lemma_retain-color} is satisfied by $\bar h$.

It remains to show that $\bar h$ satisfies~\eqref{for:vertexCut_lemma_connectivity}.
    Consider $c\in \colors$, $u\in C^{-1}(c)$, $v\notin N(u)$, and $Z\in \Gamma(u,v)$.
    We claim that, for every  $r\in C^{-1}(c)$, the following inequality holds:
\begin{equation}
\label{for:repre_star}
    \mathds{1}[r\in \rho(u)]\,\bar x_{ur}+\mathds{1}[r\in \rho(v)]\,\bar x_{vr}
    \;\leq\; \sum_{z\in Z\,:\,r\in \rho(z)}\bar x_{zr}+\bar x_{rr}.
\end{equation}
Summing~\eqref{for:repre_star} over $r\in C^{-1}(c)$ and invoking~\eqref{for:repre_3} yields
$\bar h_{uc}+\bar h_{vc}\leq \sum_{z\in Z}\bar h_{zc}+1$.

To prove~\eqref{for:repre_star}, let $V_u,V_v$ be the components of $G-Z$ containing $u$ and $v$, respectively. 
Clearly, $V_u$ and $V_v$ are disjoint since $Z\in \Gamma(u,v)$. The following cases arise:
\begin{enumerate}
    \item If $r\in\{u,v\}$, say $r=u$ without loss of generality, then~\eqref{for:repre_star} becomes $\bar x_{uu}+\bar x_{vu}\leq \sum_{z\in Z\,:\,u\in \rho(z)}\bar x_{zu}+\bar x_{uu}$, which is implied by~\eqref{for:repre_4}.
    \item If $r\in Z$, then $\bar x_{rr}\leq \sum_{z\in Z\,:\,r\in \rho(z)}\bar x_{zr}$. By~\eqref{for:repre_2}, $\bar x_{ur}\leq \bar x_{rr}$ and $\bar x_{vr}\leq \bar x_{rr}$, hence $\bar x_{ur}+\bar x_{vr}\leq 2\bar x_{rr}\leq \sum_{z\in Z\,:\,r\in \rho(z)}\bar x_{zr}+\bar x_{rr}$.
    \item If $r\in V_v\setminus\{v\}$, then $u$ and $r$ lie in different components of $G-Z$, so $Z$ is a (not necessarily minimal) $(u,r)$-vertex cut. We also know $\{u,r\}\notin E$, otherwise they cannot be separated into different components of $G-Z$. 
    Since $Z$ is a $(u,r)$-cut, there exists a minimal $(u,r)$-cut $Z'\subseteq Z$. Applying~\eqref{for:repre_4} gives $\bar x_{ur}\leq \sum_{z\in Z': r\in \rho(z)}\bar x_{zr}\leq \sum_{z\in Z: r\in \rho(z)}\bar x_{zr}$. 
    Combining this inequality and $\bar x_{vr}\leq \bar x_{rr}$, we obtain~\eqref{for:repre_star}. 
    The case $r\in V_u\setminus\{u\}$ is handled analogously.
    Moreover,  if $r$ belongs to a component of $G-Z$ other than $V_u$ and $V_v$, then $Z$ separates $r$ from $u$, and the same argument as in  the case $r\in V_v\setminus\{v\}$ applies.
\end{enumerate}
This exhausts all cases and establishes~\eqref{for:repre_star},  and therefore $\bar h$ satisfies~\eqref{for:vertexCut_lemma_connectivity}.

Regarding the objective value of~$\bar h$, it follows from  the definitions of $\hat w$ and $w$ that

\[
    \sum_{v \in V} \sum_{c \in \colors} w_{vc} \, \bar h_{vc}=
    \sum_{v\in V} w(v)\, \bar h_{v,C(v)}
    = \sum_{v\in V} w(v) \!\!\sum_{r\in C^{-1}(C(v))\cap \rho(v)}\!\! \bar x_{vr}
    = \sum_{v\in V}\sum_{r\in \rho(v)} \hat w_{vr}\, \bar x_{vr},
\]
Therefore, $\xi$ preserves the objective value.
\hfill\Halmos
\end{proof}

\subsection{Proof of~\Cref{prop:cs_in_repre}}
\begin{proof}{Proof}
Consider $\bar l\in\mathcal{R}_{\textsc{s}}$, and define
$\bar x=\phi(\bar l)$. We next show that $\bar x$ satisfies
\eqref{for:repre_1}--\eqref{for:repre_4}.

For each $u\in V$, we have  $$\sum_{v\in\rho(u)} \bar x_{uv}
    \leq
    \sum_{H\in\mathcal{S}_{\supset}(u)}
    \sum_{c\in\colors} \bar l_{Hc}
    \leq 1,$$
where the last inequality follows from~\eqref{for:connectSub_1}. Hence
$\bar x$ satisfies~\eqref{for:repre_1}.

Consider $u\in V$ and $v\in\rho(u)$. 
Since every $H \in \mathcal{S}$ contributing
to $\bar x_{uv}$ satisfies $u\in H$ and $r(H,C(v))=v \in H$, and so the same term also contributes to $\bar x_{vv}$. Hence, \eqref{for:repre_2} holds for $\bar x$.

Let $c\in\colors$.
For every $H\in \mathcal{S}(G)$ with
$H\cap C^{-1}(c)\neq\emptyset$, there exists a vertex $r(H,c)$ selected as the
representative of $c$ in $H$. 
Therefore
\[
    \sum_{v\in C^{-1}(c)} \bar x_{vv}=
    \sum_{v\in C^{-1}(c)} \sum_{\substack{H\in\mathcal{S}:\\ r(H,c)=v}}
    \bar l_{Hc} =
    \sum_{\substack{H\in\mathcal{S}:\\ H\cap C^{-1}(c)\neq\emptyset}}
    \bar l_{Hc}
    \leq
    \sum_{H\in\mathcal{S}} \bar l_{Hc}
    \leq 1,
\]
where the last inequality follows from~\eqref{for:connectSub_2}. Therefore,
\eqref{for:repre_3} holds for $\bar x$.

Consider $u\in V$, $v\in\rho(u)$
such that $\{u,v\}\notin E$, and let $Z\in\Gamma(u,v)$ and $c=C(v)$.
We investigate a subset $H \in \mathcal{S}$ with $u\in H$ and
$r(H,c)=v$. 
Since $Z$ separates $u$ and $v$, and $G[H]$ is connected, it holds that $H\cap Z\neq\emptyset$. 
Moreover, for every
$z\in H\cap Z$, we have $v\in\rho(z)$, which can be seen as follows.
If~$C(z)\neq c$ then $v\in\rho(z)$ follows directly
from the definition of $\rho(z)$; if $C(z)=c$ then $v=r(H,c)\preceq z$, which also implies $v\in\rho(z)$.
Consequently, each term contributing to $\bar x_{uv}$ contributes to at least one term
in $\sum_{\substack{z\in Z: v\in\rho(z)}} \bar x_{zv}$, and therefore \eqref{for:repre_4} is satisfied by $\bar x$. We conclude that~$\bar x\in\mathcal{R}_{\textsc{r}}$.

We now verify that the objective value is preserved.
First, by the definitions of $\phi$ and $\hat w$, we have 
\begin{align}
\sum_{u\in V}\sum_{v\in\rho(u)} \hat w_{uv}\,\bar x_{uv}
=\sum_{u\in V}w(u)\sum_{\substack{v\in\rho(u):\\ C(v)=C(u)}}\bar x_{uv}=\sum_{c\in\colors}\sum_{u\in C^{-1}(c)} w(u)
\sum_{\substack{v\in\rho(u):\\ C(v)=c}}\
\sum_{\substack{H\in\mathcal{S}_{\supset}(u):\\ r(H,c)=v}}\bar l_{Hc}. \label{eq:objvalue:subgraph}
\end{align}
Consider now $c\in\colors$, $u\in C^{-1}(c)$.
For every $H\in\mathcal{S}_{\supset}(u)$, it holds that $C(r(H,c))=c$, and \(r(H,c)\in\rho(u)\) since \(r(H,c)\) is the minimum vertex in
\(H\cap C^{-1}(c)\), and \(u\in H\cap C^{-1}(c)\). Hence
\[
\sum_{\substack{v\in\rho(u):\\ C(v)=c}}\
\sum_{\substack{H\in\mathcal{S}_{\supset}(u):\\ r(H,c)=v}}\bar l_{Hc}
\;=\;\sum_{\substack{H\in\mathcal{S}_{\supset}(u)}}\bar l_{Hc}.
\]
By substituting the previous identity in \eqref{eq:objvalue:subgraph}, we obtain
\begin{equation*}
\begin{aligned}
\sum_{u\in V}\sum_{v\in\rho(u)} \hat w_{uv}\,\bar x_{uv}&=
\sum_{c\in\colors}\sum_{u\in C^{-1}(c)}\!w(u)\!\sum_{H\in\mathcal{S}_{\supset}(u)}\bar l_{Hc}\\
&=\sum_{c\in\colors}\sum_{H \in \mathcal{S}} \bar l_{Hc} \!\sum_{u \in H\cap C^{-1}(c)} \!w(u) \\
&=
\sum_{c\in\colors}\sum_{H\in\mathcal{S}} w_{Hc}\,\bar l_{Hc},
\end{aligned}
\end{equation*}
where the third equality holds because $\sum_{u\in H\cap C^{-1}(c)} w(u)=\sum_{u\in H} w_{uc}=w_{Hc}$.  We conclude that \(\phi\) preserves the objective value.
\hfill\Halmos
\end{proof}

\subsection{Proof of~\Cref{prop:subgraph_in_flow_projection}}
\begin{proof}{Proof}

    Consider a vector $\bar l \in \mathcal{R}_{\textsc{s}}$, and let $(\bar y, \bar f) = \lambda(\bar l)$. 
    For every~$v \in V(D)\setminus\{s\}$ and $c \in \colors$, it follows immediately from the definition of~$\lambda$ that 
\begin{align}\label{eq:subgraph:inarcs}
    \sum_{H \in \mathcal{S}_\supset(v)} \bar l_{Hc} = \sum_{a \in \delta^-(v)} \sum_{\vec H_c \in \mathcal{H}_c(a)} \bar l_{Hc} = \bar y(\delta^-(v),c).
    \end{align}
    Moreover, for any $\vec H_c$ and $v \in V(\vec H_c)\setminus \{s\}$, it holds that 
    \begin{align}\label{eq:subgraph:inoutarcs}
    \sum_{a \in \delta^-(v) \cap A(\vec H_c)} n(\vec H_c,a) \bar l_{Hc} = \bar l_{Hc}  + \sum_{a \in \delta^+(v) \cap A(\vec H_c)}n(\vec H_c, a) \bar l_{Hc}
    \end{align}
     because the indegree of~$v$ in $\vec H_c$ equals~1,  and  $n(\vec H_c,b) = 1 + \sum_{a \in \delta^+(v)\cap A(\vec H_c)} n(\vec H_c, a)$, where  $\{b\}=\delta^-(v) \cap A(\vec H_c)$.

    We next use~\eqref{eq:subgraph:inarcs} and~\eqref{eq:subgraph:inoutarcs} to prove that $(\bar y, \bar f)$ satisfies all constraints in the flow formulation.
    Let $c \in \colors$.
    First observe that $\sum_{a \in \delta^+(s)} \bar y_{ac} = \sum_{a \in \delta^+(s)} \sum_{\vec H_c \in \mathcal{H}_c(a)} \bar l_{Hc} = \sum_{H \in \mathcal{S}} \bar l_{Hc} \leq 1$, where the last equation holds because $H \in \mathcal{S}$ if and only if there is an arc $a \in \delta^+(s)$ such that $\vec H_c \in \mathcal{H}_c(a)$, and  the inequality follows from~\eqref{for:connectSub_2}.
    For each $v \in V(D)\setminus \{s\}$, it holds that 
    \[\sum_{c \in \colors} \sum_{a \in \delta^-(v)} \bar y_{ac} = \sum_{c \in \colors} \sum_{a \in \delta^-(v)} \sum_{\vec H_c \in \mathcal{H}_c(a)} \bar l_{Hc} = \sum_{c \in \colors} \sum_{H \in \mathcal{S}_{\supset}(v)} \bar l_{Hc} \leq 1, \]
    where the second equation and the inequality follow from~\eqref{eq:subgraph:inarcs} and~\eqref{for:connectSub_1}, respectively.

    For each $a \in A(D)$, 
    it follows from the definition of~$\lambda$ that
    $\bar f_{ac} = \sum_{\vec H_c \in \mathcal{H}_c(a)} n(\vec H_c, a) \bar l_{Hc} \leq n\sum_{\vec H_c \in \mathcal{H}_c(a)} \bar l_{Hc} = n\bar y_{ac}$.
    Finally, for every $v \in V(D)\setminus\{s\}$, we have
    \begin{align}
    \!\!\!\sum_{a \in \delta^-(v)}\bar f_{ac} - \sum_{a \in \delta^+(v)}\bar f_{ac} &= \sum_{a \in \delta^-(v)} \sum_{\vec H_c \in \mathcal{H}_c(a)} n(\vec H_c, a) \bar l_{Hc} - \sum_{a \in \delta^+(v)} \sum_{\vec H_c \in \mathcal{H}_c(a)} n(\vec H_c, a) \bar l_{Hc} \nonumber\\
    & = \sum_{H \in \mathcal{S}_\supset(v)} \left ( \sum_{a \in \delta^-(v)\cap A(\vec H_c)}  n(\vec H_c, a) \bar l_{Hc} - \sum_{a \in \delta^+(v)\cap A(\vec H_c)} n(\vec H_c, a) \bar l_{Hc}\right) \label{eq:interchange}\\
    & = \sum_{H \in \mathcal{S}_\supset(v)} \left (  \bar l_{Hc}  + \sum_{a \in \delta^+(v) \cap A(\vec H_c)}n(\vec H_c, a) \bar l_{Hc} - \sum_{a \in \delta^+(v)\cap A(\vec H_c)} n(\vec H_c, a) \bar l_{Hc}\right)\label{eq:reduce:flow:by:one}\\
    & = \sum_{H \in \mathcal{S}_\supset(v)} \bar l_{Hc} = \sum_{a \in \delta^-(v)} \sum_{\vec H_c \in \mathcal{H}_c(a)} \bar l_{Hc} = \bar y(\delta^-(v),c),\label{eq:inarcs:summation}
    \end{align}
    where~\eqref{eq:interchange} holds by interchanging the order of the summation, \eqref{eq:reduce:flow:by:one} follows from~\eqref{eq:subgraph:inoutarcs}, and~\eqref{eq:inarcs:summation} from~\eqref{eq:subgraph:inarcs}. 

    Regarding the objective value of $(\bar y, \bar f)$, equation~\eqref{eq:subgraph:inarcs} implies that
    \begin{align*}
    \sum_{v \in V(D)\setminus\{s\}} \sum_{c \in \colors} w_{vc} \bar y(\delta^-(v),c) & = \sum_{c \in \colors}\sum_{v \in V(D)\setminus\{s\}}  w_{vc} \sum_{H \in \mathcal{S}_\supset(v)} \bar l_{Hc} =
    \sum_{c \in \colors}\sum_{H \in \mathcal{S}} \bar l_{Hc} \sum_{v \in H}  w_{vc} \\ 
    & =\sum_{c \in \colors}\sum_{H \in \mathcal{S}} w_{Hc}\bar l_{Hc}.    
\tag*{\Halmos}
    \end{align*}
\end{proof}

\section{Scale-free instance generation}
\label{Appendix:scale-free generation}

We generate scale-free network instances using the Barabási-Albert model~\citep{barabasi1999emergence}. In this model, the graph is generated using an iterative procedure as follows. 
At each iteration, a new vertex is added and connected to $d$ existing vertices, chosen with a probability proportional to their degree. 
Formally, the probability that the new vertex is linked to~$v$ is
\[
p_v = \frac{\deg(v)}{\sum_u \deg(u)},
\]
where $\deg(v)$ is the degree of~$v$ and the sum is over all existing vertices~$u$.

For the unweighted scale-free instances, the initial coloring is assigned uniformly at random to the vertices. We generate three graphs for each combination of
$n \in \{60,90,120\}$,
$k \in \{\frac{n}{15},\frac{n}{10},\frac{n}{5},\frac{n}{3},\frac{n}{2}\}$,
and $d \in \{2,4,6\}$.  
This gives $135$ unweighted scale-free instances.

For the weighted scale-free instances, the weight of each vertex is set equal to its degree. We generate three graphs for each combination of
$n \in \{90,120,150\}$,
$k \in \{\frac{n}{15},\frac{n}{10},\frac{n}{5},\frac{n}{3},\frac{n}{2}\}$,
and $d \in \{2,4,6\}$,
resulting in $135$ weighted scale-free instances.

The weighted and unweighted instances are derived from the same base graphs for
each number of vertices.  In the scale-free setting, the weighted instances also include larger graphs, as vertex weights tend to make the problem easier to solve.

\end{document}